\documentclass{amsart}
\usepackage{vmargin,amsmath,amsfonts,amssymb,amsthm,xcolor,amscd,latexsym,xspace,enumerate}
\usepackage[pagebackref, colorlinks=true,linkcolor=blue,citecolor=red]{hyperref}
\usepackage{hyperref}
\usepackage{pdfpages}
\usepackage{tikz-cd}
\usetikzlibrary{arrows}
\usetikzlibrary{intersections}
\usepackage{pgfplots}
\usetikzlibrary{calc,3d,shapes, pgfplots.external, intersections}
\usepackage{pgfplots}
\newtheorem{theorem}{Theorem}[section]
\newtheorem{corollary}[theorem]{Corollary}

\newtheorem{proposition}[theorem]{Proposition}
\theoremstyle{definition}
\newtheorem{definition}[theorem]{Definition}
\newtheorem{remark}[theorem]{Remark}
\newtheorem{conjecture}[theorem]{Conjecture}

\newtheorem{problem}[theorem]{Problem}
\newtheorem{example}[theorem]{Example}

\numberwithin{equation}{section}
\usepackage{stmaryrd}
\usepackage{mathtools}
\usepackage{xcolor,soul}
\sethlcolor{green}

\newcommand{\be}{\begin{equation}}
\newcommand{\en}{\end{equation}}
\newcommand{\bea}{\begin{eqnarray}}
\newcommand{\ena}{\end{eqnarray}}
\newcommand{\beano}{\begin{eqnarray*}}
\newcommand{\enano}{\end{eqnarray*}}
\newcommand{\bee}{\begin{enumerate}}
\newcommand{\ene}{\end{enumerate}}

\newcommand{\mc}{\mathcal}

\newcommand{\D}{{\mc D}}

\newcommand{\1}{1 \!\! 1}
\title[Some Consequences of the  Grunewald-O'Halloran Conjecture...]{Some Consequences of the  Grunewald-O'Halloran Conjecture for Pseudoquonic  Operators}

\author[F. Bagarello]{Fabio Bagarello}
\address{ Fabio Bagarello \endgraf
Dipartimento di Ingegneria \endgraf 
Universit\`a di Palermo\endgraf
Viale delle Scienze 90128  \endgraf
 Palermo, Italy\endgraf 
}

\author[Y. Bavuma]{Yanga Bavuma}
\address{ Yanga Bavuma  \endgraf
Department of Mathematics and Applied Mathematics\endgraf 
University of Cape Town\endgraf
Private Bag X1, Rondebosch 7701\endgraf 
Cape Town, South Africa\endgraf
}

\author[F.G. Russo]{Francesco G. Russo}
\address{Francesco G. Russo \endgraf
School of Science and Technology, University of Camerino\endgraf
Via Madonna delle Carceri 9, Camerino, Italy\endgraf
and\endgraf
Department of Mathematics and Applied Mathematics, University of the Western Cape\endgraf
Private Bag X17, 7535,  Bellville, South Africa\endgraf
and\endgraf
Department of Mathematics and Applied Mathematics, University of Cape Town\endgraf
Private Bag X1, Rondebosch 7701, Cape Town, South Africa\endgraf
}

\begin{document}

\begin{abstract} Investigating a recent positive solution of a conjecture of Grunewald and O'Halloran for complex finite dimensional nilpotent Lie algebras, we are in the position to find results of existence and uniqueness for the construction of complex nilpotent Lie algebras of arbitrary dimension via pseudobosonic operators. We involve the so-called theory of the deformation of Lie algebras of Gerstenhaber, in order to prove our main results. There isn't a generalized version of the Grunewald-O'Halloran Conjecture when we consider pseudoquonic operators, which specialize to pseudobosonic operators in many cirumstances. Therefore we  prove a result of existence (and a direct construction) of pseudobosonic $O^*$-algebras of operators, but  leave open the problem of the uniqueness of the construction.
\end{abstract}

\maketitle

\tableofcontents

\section{Introduction}\label{sec1}

From the first works of von Neumann \cite{vN39}, it was clear  that suitable functional operators played a fundamental role in the description of several dynamical systems in quantum mechanics. Just to give an idea, we may consider  bosons and fermions, which have been studied for a long time and by many authors. These elementary particles form two significant classes,  which are involved almost everywhere in quantum mechanics. At a first level, we may think of fermions as elementary particles which are associated to the structure of matter and can be formalized when an appropriate framework of functional analysis is constructed. At a second level, we may involve  a complex  Hilbert space   and some appropriate   lowering and  raising operators, which indeed lower or raise the eigenvalues of  corresponding eigenstates of a dynamical system. Roughly speaking, this is the main idea to settle  a classical mathematical model, which can help us to understand  properties of a dynamical system. For instance, quantum harmonic oscillators or two level atoms of hydrogen, can be described by fermionic operators in this way, see \cite{tripf}. 

Just to illustrate more the motivation for the representation theory in mathematical physics, we note that it often happens that fermionic operators are  formalized as a pair of functional operators, say $A$ and $B$, where the lowering operator is $A$  and the raising operator is $B$,  satisfy a deformed version of the \textit{ canonical anticommutation relations} (CAR) 
\begin{equation} \label{pseudofermions}
    \{A,B\}=A B +BA  = \mathbb{I},  \qquad \mbox{and} \qquad \{A,A\}=\{B,B\}=0.
\end{equation}
Here $\mathbb{I}$ denotes the identity operator and  $B \neq A^\dag$ a priori, that is, $B$ is not necessarily the adjoint operator $A^\dag$ of $A$. In particular, if this happens, that is, if $B = A^\dag$  and    \eqref{pseudofermions} are satisfied,  $A$ is a fermionic operator (of adjoint $A^\dag$). The terminology is quite standard among the theory of ladder operators, and  in fact one can give physical interpretations of the rules of algebraic nature such as \begin{equation} \label{pseudofermionsbis}
    \{A,A\}=0 \Longleftrightarrow  A^2=0 \qquad \mbox{and} \qquad 
    \{B,B\}=0 \Longleftrightarrow  B^2=0
\end{equation}
in terms of conservation laws or  physical properties.
Moreover one can introduce the operator  $    N = BA$ and give additional interpretations as the Hamiltonian of a dynamical system which satisfy CAR.  The relevance of these  notions was largely explored  in algebraic quantum field theory  and in Lie theory, see \cite{bag2022book, bender1}. Of course, one can generalize in several directions the aforementioned mathematical models, introducing the so called pseudobosonic operators and the pseudoquonic operators, see \cite{bagcohsta, bagquons,  eremel, fiv,  bagrus2018, bagrus2020, gre}. There are in fact  large areas of mathematical physics which are currently devoted to investigate ladder operators with a special focus on nonselfadjoint Hamiltonians in dynamical systems and PT-symmetries.

In the present paper we continue some of our previous investigations in \cite{bagrus2018, bagrus2019, bagrus2020} in connection with the so-called ``deformation theory of Lie algebras'' of Gerstenhaber \cite{gerstenhaber1, gerstenhaber2, gerstenhaber3, gerstenhaber4}. This has been shown to be a powerful tool for the classification of several algebraic and geometric objects, but is receiving attention only in recent years in the area of quantum mechanics. Our main results are placed in Sections \ref{sec4} and \ref{sec5}, in fact they are Theorems \ref{main1}, \ref{main2} and \ref{main3}, illustrating a first connection of the theory of pseudobosonic operators and pseudoquonic operators with the deformation theory of Gerstenhaber. There are a series of open questions and problems, which we found along our investigations placed in Section \ref{sec6}. These will be addressed in successive contributions, after we develop a cohomological theory of q-deformed Heisenberg algebras, which is available at the moment, but not very wide and appropriate for  the investigations which we have done in connection with pseudobosonic operators. Section \ref{sec2} develops in addition to Appendix I of Section \ref{appendix1} some minimal notions which one needs to report, in order to present the main ideas of the deformation theory of finite dimensional complex nilpotent Lie algebras. There are many connections with homological algebra, algebraic geometry, differential geometry and algebraic topology in these sections. Then we shortly report the main ideas of the theory of the pseudobosonic operators in Section \ref{sec3}, having all that we need for the proofs of Theorems \ref{main1} and \ref{main2} in Section \ref{sec4}. Here essentially we focus only on pseudobosonic operators in connection with some conjectures of Grunewald and others \cite{grunewald1, vergne}. Later on we develop a more general construction, where we involve pseudoquonic operators in Section \ref{sec5}, and we find that most of the results which  have been seen for pseudobosonic operators are in agreement with some special situations where pseudoquonic operators appear, see Theorem \ref{main3}. Unfortunately, the analogies are not always perfect and some careful distinctions should be made, since q-deformed Heisenberg algebras replace the role of Lie algebras when one wants to investigate pseudoquonic operators (some elementary notions are captured in Appendix II of Section \ref{qdeformedappendix}). These structures are much more general and a series of open problems appear for their complete classification, so that
 a deformation theory for q-deformed Heisenberg algebras (in the sense of Gerstenhaber) isn't well developed at the moment. We observe this in Section \ref{sec6}.

\section{Grunewald-O'Halloran Conjecture and Vergne's Conjecture: a short review}\label{sec2}

The theory of finite dimensional Lie algebras is well developed in connection with the mathematical physics, since Lie algebras are involved in several aspects of theoretical physics from the origins of the quantum mechanics, see \cite{gokh, sergey2, snobwin}.  Despite finite dimensional Lie algebras on fields of characteristic zero being among the most studied Lie algebras, there are still some interesting  open problems, which correlate their structure  with important applications. In fact, the presence of conservation laws can be detected when one formalizes certain dynamical systems using appropriate   ``Canonical Commuting Relations'' (CCR) and  ``Canonical Anticommuting Relations'' (CAR), see \cite{bagcohsta, bagquons, bagrus2018}. Grunewald and O'Halloran \cite[Conjecture]{grunewald1} formulated the idea of degeneration of Lie bracket (see Definition \ref{degenerate} below) and noted that:

\begin{conjecture}[Grunewald--O'Halloran Conjecture]\label{goh} Every  complex  finite dimensional nilpotent Lie algebra (of dimension more than two) is the degeneration of another  Lie algebra. \end{conjecture}

Conjecture \ref{goh} has been proved by  Herrera-Granada and Tirao \cite{herrera1} in low dimensional cases.

\begin{theorem}[See \cite{herrera1}, Theorem 2]\label{niceconfirmation} Conjecture \ref{goh} is true when the dimension is at most seven.
\end{theorem}

This confirmed previous investigations of Seeley and Yau in \cite{seeley1, seeley2}, but also the theory of Gerstenhaber \cite{gerstenhaber1, gerstenhaber2, gerstenhaber3, gerstenhaber4} who introduced the ``deformation theory of Lie algebras'' via cohomological techniques. We inform the reader that one should be careful with the meaning of the words ``degeneration'' and ``deformation'', see \cite{gerstenhaber1, grunewald1, grunewald2} or  Definitions \ref{deformation} and \ref{lindef} below. Note also that Dixmier and others \cite{ancogoze1,  dixmier} gave results of structure for  complex nilpotent Lie algebras of low dimension without using the deformation theory, see also \cite[Chapters 16, 17, 18, 19]{snobwin}.  

Appendix I  of Section \ref{appendix1}  has been written, in order to offer feedback and preliminary notions of homological nature to the original ideas of the deformation theory of Lie algebras. For larger dimensions, it is in fact natural to  assume that there are complex finite dimensional Lie algebras with nontrivial semisimple derivations and with this hypothesis Conjecture \ref{goh} is valid, see \cite[Theorem 1]{herrera1}. Note that  Remark \ref{derivations} gives the formal definition of \textit{nontrivial semisimple derivation}.  What happens if we don't have nontrivial semisimple derivations and the dimension is $ \ge 8$ ? It is still possible to have a positive answer, as indicated by \cite{herrera2, herrera3}, but one must consider  filiform Lie algebras, or alternative assumptions of structural nature, since otherwise a positive solution of the  Grunewald--O'Halloran Conjecture is not guaranteed.

Another well-known conjecture, which is interesting the scientific community of mathematicians and physicists working in Lie theory, can be found in \cite{vergne} and is due to Vergne:

\begin{conjecture}[Vergne's Conjecture]\label{vc} There are no rigid complex nilpotent Lie algebras  in the algebraic variety $\mathrm{L}_n$ of complex Lie algebras of dimension $n$. \end{conjecture}

Conjecture \ref{vc} is more sophisticated to discuss here, since it requires some familiarity with algebraic geometry, Lie theory and algebraic topology. For instance, one can show that the Grunewald--O'Halloran Conjecture is stronger than the Vergne's Conjecture when we consider finite dimensional nilpotent Lie algebras, but they are not equivalent in general.

%Let's see more formally the  meaning of what we have just written:

\begin{definition}[See \cite{hartshorne}, Definition, Chapter I, \S 1]\label{av}  Consider a polynomial $f_i$ in  $n$ complex variables $x_1$, $x_2$, ..., $x_n$  with complex coefficients and roots 
   \begin{equation} f_i^{-1}(0)=\{ (x_1, \ldots, x_n) \in \mathbb{C}^n \mid f_i (x_1, \ldots, x_n) = 0\},
   \end{equation}
where $i \in \{1,2, \ldots,m\}$  and $m \in \mathbb{N}$. Now consider  $S=\{f_1, \ldots,f_m \} \subseteq \mathbb{C}[x_1, \ldots, x_n]$, where $\mathbb{C}[x_1, \ldots, x_n]$ is the ring of  polynomials in $n$ variables with complex coefficients, and  the roots
\begin{equation} N(S)=\{(x_1, \ldots, x_n) \in \mathbb{C}^{n} \mid f_i(x_1, \ldots, x_n)=0,  \ \ \forall f_i \in S \} =\bigcap^m_{i =1} f_i^{-1}(0).
\end{equation}
A subset $T$ of $\mathbb{C}^n$ is called \textit{affine variety}, if $T = N(S)$, that is, if there exists some set $S$ of finitely many polynomials of $\mathbb{C}[x_1, \ldots, x_n]$ for which their roots $N(S)$ agree with $T$. Now  the set of $\{f^{-1}_i(0) \mid 1\le i \le m \}$ forms a basis for the closed sets of a topology, which  is known as  \textit{Zariski topology} on $\mathbb{C}^n$, and   $T$ possesses  the induced Zariski topology from $\mathbb{C}^n$. 
\end{definition}

Following \cite{grunewald1} and \cite[Appendix 3]{hofmor}, we may define the usual \textit{wedge product} $(x_1, \ldots, x_n) \wedge (y_1, \ldots, y_n)$ of two complex vectors $(x_1, \ldots, x_n)$ and $ (y_1, \ldots, y_n)$ of $\mathbb{C}^n$ and introduce
\textit{complex exterior vector space} \begin{equation}\wedge^2 \mathbb{C}^n = \mathbb{C}^n \wedge \mathbb{C}^n =\{ (x_1, \ldots, x_n) \wedge (y_1, \ldots, y_n) \mid x_1, y_1, \ldots, x_n,y_n \in \mathbb{C}\} 
\end{equation}
which is in fact an algebra (called \textit{exterior algebra}). Consequently, we have the vector space \begin{equation}
\mathrm{Hom}(\wedge^2 \mathbb{C}^n,\mathbb{C}^n)= \{ \varphi : \wedge^2 \mathbb{C}^n \mapsto \mathbb{C}^n  \mid \varphi \ \mbox{is a linear homomorphism of complex vector spaces}\}
\end{equation} and in this situation  one can see that: 

\begin{remark} $\mathrm{Hom}(\wedge^2 \mathbb{C}^n,\mathbb{C}^n)$  possesses the structure of affine variety as per Definition \ref{av} and    the set $\mathrm{L}_n $ containing the alternating bilinear maps which satisfy the Jacobi identity is a subset of $\mathrm{Hom}(\wedge^2 \mathbb{C}^n,\mathbb{C}^n)$, see details in \cite{grunewald1, grunewald2, rem1}. In particular,   $\mathrm{L}_n $ possesses a structure of affine variety. \end{remark}

We shall recall another important observation, which is present in \cite{grunewald1, herrera1, herrera2, herrera3, rem1, vergne}.

\begin{remark}\label{liebracketsiso}
Among all  complex Lie algebras of finite dimension, a Lie algebra $\mathfrak{g}$  is uniquely determined by its Lie bracket $\mu$. Therefore it is equivalent either to assign a finite dimensional complex Lie algebra or to assign a Lie bracket. \end{remark}

Thanks to Remark \ref{liebracketsiso}, Grunewald and other authors \cite{burde1, burde2, gerstenhaber1, gerstenhaber2, grunewald1, herrera1, herrera2, herrera3} call $\mathrm{L}_n$  the \textit{affine variety of complex Lie algebras of dimension} $n$, or also,  the \textit{affine variety of Lie brackets} $\mu$ \textit{on} $\mathbb{C}^n$, since one can indistinctly refer to $\mathfrak{g}$ or to  $\mu$ without ambiguities. Again from \cite{grunewald1}, we know that the general linear group $\mathrm{GL}_n(\mathbb{C})$ of dimension $n$ acts on $\mathrm{L}_n$ via 
\begin{equation}\label{deformation} (g,\mu) \in \mathrm{GL}_n(\mathbb{C}) \times \mathrm{L}_n \mapsto g \cdot \mu (x,y)=g(\mu(g^{-1}x,g^{-1}y)) \in \mathrm{L}_n, 
\end{equation}
producing  \begin{equation}\label{orbitdef} \mbox{orbits} \ \ \mathcal{O}(\mu)= \mathrm{GL}_n(\mathbb{C}) \cdot \mu  \ \ \  \mbox{in} \ \  \mathrm{L}_n \ \  \ \ \ \ \mbox{and} \ \  \ \ \ \mbox{stabilisers} \ \   \ {\mathrm{GL}_n(\mathbb{C})}_\mu \ \ \mbox{in} \ \  \mathrm{GL}_n(\mathbb{C}).\end{equation} The first are the isomorphism classes of $\mu$, so we get the orbit space
\begin{equation}
\frac{\mathrm{L}_n}{\mathrm{GL}_n(\mathbb{C})} =\bigcup_{\mu \in \mathrm{L}_n} \mathcal{O}(\mu)
\end{equation}
while the second are the elements of $\mathrm{GL}_n(\mathbb{C})$ fixing $\mu$ under the action \eqref{deformation}.  It is also useful to mention here that $\mathcal{O}(\mu)$ possesses the structure of differentiable homogeneous manifold in  $\mathrm{L}_n$, see \cite{gokh, grunewald1, grunewald2, hofmor} and Appendix I in Section \ref{appendix1}.

\begin{definition}[Degeneration and Rigidity, see \cite{grunewald1, grunewald2, rem1}]\label{degenerate} A Lie bracket $\mu$ is said to \textit{degenerate} to $\lambda$ if $\lambda \in \partial \ \mathcal{O}(\mu) $, that is, if $\lambda$ belongs to the  boundary  of the orbit $\mathcal{O}(\mu)$ with respect to the Zariski topology. Instead we say that $\mu$ is \textit{rigid} if $\mathcal{O}(\mu)$ is open in the Zariski topology.
\end{definition}

The absence of degenerations is described topologically by the condition of rigidity, while the presence of rigidity in Definition \ref{degenerate} shows the absence of degenerations.
 The most famous example of rigid Lie algebra is given by $\mathfrak{sl}_2(\mathbb{C})$, that is, by the special linear Lie algebra  of  $2$-by-$2$ with complex coefficients. Note that Definition \ref{degenerate} means that  $\lambda \in \overline{\mathcal{O}(\mu) } \setminus \mathcal{O}(\mu)$, that is, $\lambda$ belongs to  the complement of $\mathcal{O}(\mu)$ in the closure $\overline{\mathcal{O}(\mu) }$ with respect  to the Zariski topology.

\begin{definition}[See \cite{grunewald1}, Definition 1.1]\label{lindef} A \textit{linear deformation of the Lie bracket} $\mu$ is a family of Lie brackets $\mu_t$ with $t$ nonzero complex number such that
\begin{equation}\label{ld} \mu_t=\mu + t\phi, 
\end{equation}
where $\phi  $ is a skew-symmetric bilinear form on $\mathbb{C}^n$ (i.e.: an element of $\mathrm{L}_n$).
\end{definition}

More generally, one can introduce nonlinear deformations which specialize to Definition \ref{lindef}, beginning with  $\phi_i \in \mathrm{Hom}(\wedge^2 \mathbb{C}^n,\mathbb{C}^n)$  skew-symmetric bilinear form for all $i=1,2,\ldots,k$ with $k$ positive integer and calling \textit{k-deformation of the Lie bracket} $\mu$ the sum
\begin{equation}\label{nld} \mu_t=\mu + t\phi_1 + t^2\phi_2 + \ldots + t^k\phi_k=\mu + \sum^k_{i=1} t^i \phi_i. 
\end{equation}

We  refer to \eqref{nld} as \textit{nonlinear deformation of the Lie bracket} $\mu$ whenever  $\sum^k_{i=2} t^i \phi_i \neq 0$.

\begin{remark}\label{subtledifference} It is crucial to observe that Definitions \ref{degenerate} and \ref{lindef} are different a priori. However  the orbit closure   $ \overline{\mathcal{O}(\mu)} $ has been well described in \cite{grunewald1, grunewald2} and it is possible to see that every degeneration $\mu_1$ of a Lie algebra give rise to a deformation $\mu_2$. Thus the existence of degenerations implies the existence of nontrivial deformations. Roughly speaking, deformations are special types of degenerations, but there are degenerations which are not necessarily deformations. In any case, if  $\mu$ is rigid, then we  have neither degenerations nor deformations. We omit details, since they  can be found in \cite{grunewald1, grunewald2}, 
\end{remark}

\begin{problem}A complete classification of finite dimensional complex Lie algebras with quadratic deformations of the Lie bracket, i.e.:  $\mu_t=\mu + t \phi_1 + t^2 \phi_2$, seems to be still an open problem.  \end{problem}

Several interpretations of physical nature have been proposed for  \eqref{nld}. For instance,  we may think that a ``deformation'' of a mathematical object is a family of the same kind of objects depending on  parameters. An important example is offered by the infinitesimal deformations of an algebra of a given type (so not necessarily finite dimensional complex Lie algebras) and these are known to be parametrized by a second cohomology of the algebra, see Definition \ref{lindef} and Remarks \ref{derivation2}, \ref{derivation3} and \ref{crucialobservation}.

\begin{remark} Deformations of Lie algebras may be controlled by cohomological restrictions, in fact the main idea is to construct new objects starting from a given object and to infer some of its original properties and this is typical of extension theory in abstract algebra.  It should also be mentioned that some mathematical formulations of the notion of  quantization are based on the algebra of observables and consist in replacing the classical algebra of observables  by a noncommutative one constructed via  deformations of classical algebras. This idea can be found in the works of Kontsevich \cite{kont1, kont2} and was used to solve a longstanding problem in mathematical physics, that is, every Poisson manifold admits formal quantization which is canonical up to  equivalence.
\end{remark}

\begin{remark}\label{jacobiidentity} It is useful to note that \eqref{ld} satisfies the Jacobi identity, since for $x,y,z$ belonging to a finite dimensional complex Lie algebra $\mathfrak{g}$ with Lie bracket $\mu$ and $\phi $ as in Definition \ref{lindef}, we get for a nonzero complex number $t$ that
\begin{equation}\label{eq1}
\mu_t(x, \mu_t(y,z))=\mu(x, \mu(y,z)) + t \phi(x, \phi(y,z));
\end{equation}
\begin{equation}\label{eq2}
\mu_t(y, \mu_t(z,x))=\mu(y, \mu(z,x)) + t \phi(y, \phi(z,x));
\end{equation}
\begin{equation}\label{eq3}
\mu_t(z, \mu_t(x,y))=\mu(z, \mu(x,y)) + t \phi(z, \phi(x,y));
\end{equation}
and so we find the Jacoby identity
\begin{equation}    \mu_t(x, \mu_t(y,z))  + \mu_t(y, \mu_t(z,x)) + \mu_t(z, \mu_t(x,y))    
\end{equation}
\[=\mu(x, \mu(y,z)) + t \phi(x, \phi(y,z)) + \mu(y, \mu(z,x)) + t \phi(y, \phi(z,x)) + \mu(z, \mu(x,y)) + t \phi(z, \phi(x,y))\]
\[=\underbrace{\mu(x, \mu(y,z)) + \mu(y, \mu(z,x)) + \mu(z, \mu(x,y))}_{0} + t \underbrace{(\phi(x, \phi(y,z)) + \phi(y, \phi(z,x)) + \phi(z, \phi(x,y)))}_0 = 0.\]
%Of course, a similar argument can be extended to \eqref{nld}, but we omit details for brevity. 
\end{remark}

We may summarize the argument above in the following statement:

\begin{proposition}[Preservation Property of Jacobi Identity]\label{preservationproperty1} Both linear deformations and nonlinear deformations of Lie brackets  preserve the Jacobi identity.
\end{proposition}

\begin{proof} Consider $ \mu_t=\mu + t\phi_1 + t^2\phi_2 + \ldots + t^k\phi_k$ nonlinear deformation of the Lie bracket $\mu$ as per \eqref{nld} with $k \ge 1$. We do induction on $k$. For $k=1$, we get linear deformations and the result follows from Remark \ref{jacobiidentity}. Assume $k \ge 2$ and that the result is true for $k-1$. Then  for $x,y,z$ belonging to a finite dimensional complex Lie algebra $\mathfrak{g}$ we get
\begin{equation}    \mu_t(x, \mu_t(y,z))  + \mu_t(y, \mu_t(z,x)) + \mu_t(z, \mu_t(x,y))    
\end{equation}
\[=\underbrace{\mu(x, \mu(y,z)) + \sum^{k-1}_{i=1} t^i \phi_i(x, \phi_i(y,z))}_{0 \ \mbox{by induction hypothesis}} + t^k \phi_k(x, \phi_k(y,z))\]
\[+ \underbrace{\mu(y, \mu(z,x)) + \sum^{k-1}_{i=1} t^i \phi_i(y, \phi_i(z,x))}_{0 \ \mbox{by induction hypothesis}} + t^k \phi_k(y, \phi_k(z,x))\]
\[ + \underbrace{\mu(z, \mu(x,y)) + \sum^{k-1}_{i=1} t^i \phi_i(z, \phi_i(x,y))}_{0 \ \mbox{by induction hypothesis}} + t^k \phi_k(z, \phi_k(x,y))\]
\[=t^k \ \ \underbrace{(\phi_k(x, \phi_k(y,z)) + \phi_k(y, \phi_k(z,x)) + \phi_k(z, \phi_k(x,y))}_{0  \ \mbox{ since} \ \phi_k \ \mbox{satisfies the Jacobi identity}}) = 0.\]
  \end{proof}

We are going to focus on Definition \ref{lindef}, so mainly on linear deformations of Lie brackets. In Appendix I, which is placed in Section \ref{appendix1} below, we discuss  the  Schur multiplier $M(\mathfrak{g})=Z^2(\mathfrak{g})/B^2(\mathfrak{g})$ for a complex finite dimensional nilpotent Lie algebra $\mathfrak{g}$, see Remark \ref{derivation2}; this notion along with Remarks \ref{derivation3} and \ref{crucialobservation} shows that Definitions \ref{lindef},  \ref{extensions} and \ref{ce2} are strongly related. In fact, changing Lie bracket, at least for linear deformations, is equivalent to consider extensions of Lie algebras.

\begin{proposition}[Influence of Schur Multipliers on Lie Brackets] \label{ism}Consider a complex finite dimensional  nilpotent Lie algebra $\mathfrak{g}$ with Lie bracket $\mu$ and  Schur multiplier $M(\mathfrak{g})=Z^2(\mathfrak{g})/B^2(\mathfrak{g})$. Then  $\mu_t$ is a linear deformation of $\mu$ if and only if $\phi \in Z^2(\mathfrak{g})$. In particular, there are no linear deformations of $\mu$ when $M(\mathfrak{g})$ is trivial.
\end{proposition}

\begin{proof} 
If $\phi \in Z^2(\mathfrak{g})$, then by definition  $\phi : (x , y) \in  \mathfrak{g} \times \mathfrak{g} \mapsto \phi(x , y)  \in  \mathfrak{g}$ satisfies the Jacobi identity and $\phi$ is  alternating bilinear, hence $\phi$ is a Lie bracket on $\mathfrak{g}$ and we may form $\mu_t$ as per Definition \ref{lindef}. Viceversa, assume that $\mu_t$ is given and $\phi \in \mathrm{Hom}(\wedge^2 \mathbb{C}^n, \mathbb{C}^n)$ is alternating bilinear and satisfies the Jacobi identity. Then by definition 
 $\phi \in Z^2(\mathfrak{g}) $ so the result follows.
\end{proof}

It should be mentioned that the final part of Proposition \ref{ism} follows independently from  Nijenhuis-Richardson Theorem (see  \cite[Theorem 8]{rem1}), which ensures that complex finite dimensional Lie algebras with trivial Schur multiplier must be rigid. In other words, 
the Schur multiplier shouldn't be trivial in our investigations and this is really important.

\begin{example}Again it is useful to consider the special linear Lie algebra $\mathfrak{sl}_2(\mathbb{C})$, which we may also define by a vanishing condition on the traces, that is, by \begin{equation}\mathfrak{sl}_2(\mathbb{C})=\{x \in \mathfrak{gl}_2(\mathbb{C})   \ | \ \mathrm{tr}(x)=0\}.\end{equation}  This is simple, according to  \cite[Definition 7.8.1]{weibel}. If $x^t$ denotes the transpose of $x \in \mathfrak{gl}_2(\mathbb{C})$, then we may also consider the special orthogonal Lie algebra \begin{equation}\mathfrak{so}_3(\mathbb{R})=\{x \in \mathfrak{gl}_2(\mathbb{R})   \ | \ x+x^t=0\}.\end{equation}
 Note that we have the scalar product on $\mathfrak{sl}_2(\mathbb{C})$  \begin{equation} \langle y, z \rangle= \mathrm{tr}(yz)\end{equation}  and the adjoint map \begin{equation}\mathrm{ad} \ : \ x \in \mathfrak{sl}_2(\mathbb{C}) \longmapsto \mathrm{ad}(x) \in \mathrm{Der}(\mathfrak{sl}_2(\mathbb{C}))\end{equation}
acts by  preserving the scalar product. Of course,  $ \mathrm{Der}(\mathfrak{sl}_2(\mathbb{C}))$ denotes the set of all derivations in  $\mathfrak{sl}_2(\mathbb{C})$, see Appendix I below, and we have for all $x,y,z \in \mathfrak{sl}_2(\mathbb{C})$ that \begin{equation}\langle \mathrm{ad}(x) y, z \rangle + \langle y, \mathrm{ad}(x)z \rangle=0.\end{equation}
Therefore the adjoint action induces a Lie homomorphism from $\mathfrak{sl}_2(\mathbb{C})$ to $\mathfrak{so}_3(\mathbb{R})$, which is in fact an  isomorphism such that $\mathfrak{sl}_2(\mathbb{C}) \simeq \mathfrak{so}_3(\mathbb{R}).$ Now $M(\mathfrak{sl}_2(\mathbb{C}))$ is trivial and $\mathfrak{sl}_2(\mathbb{C})$ is a finite dimensional nonnilpotent complex Lie algebra. One can see that $\mathfrak{sl}_2(\mathbb{C})$ does not possess any linear deformation $\mu_t$ and this is in agreement with the fact that its Schur multiplier is trivial.  In fact the only complex rigid Lie algebra of dimension 3 is  $\mathfrak{sl}_2(\mathbb{C})$. On the other hand, one can see that the solvable (nonnilpotent) Lie algebra $\mathfrak{r}+\mathbb{C}$ of dimension three, where $\mathfrak{r}$ is the 2-dimensional solvable Lie algebra is not the degeneration of any other Lie algebra, see  \cite{burde1, burde2}. This is to support the fact that the assumption of being nilpotent in Proposition \ref{ism} is essential.
 %Therefore we have evidences that the two conjectures are different. 
\end{example}

\section{Previous results concerning dynamical systems with pseudobosons}\label{sec3}

We report some notions on  the so-called $\D$-pseudo bosons (briefly, $\D$-PBs), which can be found in \cite{bag2022book,  bagrus2018}.  These are ladder operators, which allow to formalize properly certain dynamical systems where one can recognize PT-symmetries, see \cite{bagcohsta, bender1}.

Let $\mathsf{H}$ be a given Hilbert space with scalar product $\left< \ , \ \right>$ and related norm $\| \ \|$. As we know, there are essentially two prototypes of Hilbert spaces (up to isomorphisms), one is finite dimensional and another is infinite dimensional. For instance, in case of finite dimensional real Hilbert spaces, the first type  can be visualized with $\mathbb{R}^n$ with the usual norm. The second type is a functional space of infinite dimension.
If $C^\infty(\mathbb{R})$ denotes the real vector space of the smooth functions with support on $\mathbb{R}$ and ${\| \  \|}_2$ the usual $L^2$-norm, a prototype of infinite dimensional real Hilbert space is given by \begin{equation}\mathsf{H}=\{f(x) \mbox{ is Lebesgue-measurable on } \mathbb R \ | \ {\| f \|}_2 < \infty \},
\end{equation}but for our scopes it is useful to consider also another functional space, namely
\begin{equation}\mathcal{S}(\mathbb R)=\{f(x) \in C^\infty(\mathbb{R})  \  |  \  \lim_{|x|,\infty}|x|^kf^{(l)}(x)=0, \, \forall k,l\geq0  \ \},\label{sr}
\end{equation} 
which denotes the set of  smooth functions which decrease, together with their derivatives, faster than any inverse power (see \cite{bag2022book}). The space
$\mathcal{S}(\mathbb R)$ is known as the \textit{Schwartz space}. This functional space is very interesting for several reasons; for instance, $\mathcal{S}(\mathbb R)$ is neither a Hilbert space, nor a Banach space, in fact it is a nuclear space (in the sense of Grothedieck), see \cite{aitbook, chri}. On the other hand, one can see that the Fourier transform is a linear isomorphism on $\mathcal{S}(\mathbb R)$. Still more interesting it is to check that $\mathcal{S}(\mathbb R)$ is a locally convex Frech\'et space, moreover $\mathcal{S}(\mathbb R)$ is complete and Hausdorff, but not metrizable. %More difficult, it is to check that in this space any sequence converges in the strong dual topology if and only if it converges in the weak* topology.

In this situation we consider   two operators $A$ and $B$
on $\mathsf{H}$, with domains $D(A)$ and $D(B)$ respectively, $A^\dagger$ and $B^\dagger$ their adjoint, and let $\D$ be a dense subspace of $\mathsf{H}$
such that $A^\sharp\D\subseteq\D$ and $B^\sharp\D\subseteq\D$, where $X^\sharp$ is either $X$ or $X^\dagger$. Of course, $\D\subseteq D(A^\sharp)$
and $\D\subseteq D(B^\sharp)$.

\begin{definition}[CCR for Pseudobosonic Operators, see \cite{bag2022book}]\label{carrules}
The operators $(A,B)$ are $\D$-pseudo bosonic  if, for all $f\in\D$, we have
\be
A\,B\,f-B\,A\,f=f.
\label{A1}\en
\end{definition}

\vspace{2mm}

Our  working assumptions are recalled from \cite{bag2022book, bagrus2018}:

\vspace{2mm}

{\bf Assumption $\D$-pb 1.}  There exists a nonzero $\varphi_{ 0}\in\D$ such that $A\,\varphi_{ 0}=0$.

\vspace{1mm}

{\bf Assumption $\D$-pb 2.}  There exists a nonzero $\Psi_{ 0}\in\D$ such that $B^\dagger\,\Psi_{ 0}=0$.

\vspace{2mm}

Then, if $(A,B)$ satisfy Definition \ref{carrules}, it is obvious that \begin{equation}\varphi_0\in D^\infty(B):=\bigcap_{k\geq0}D(B^k) \ \ \mbox{and} \ \ \Psi_0\in D^\infty(A^\dagger)
\end{equation} so
that for all $n \ge 0$ the vectors 
\be \varphi_n:=\frac{1}{\sqrt{n!}}\,B^n\varphi_0 \ \ \mbox{and} \ \ \Psi_n:=\frac{1}{\sqrt{n!}}\,{A^\dagger}^n\Psi_0 \label{A2a}\en
 can be defined and they all belong to $\D$. As a consequence, they belong to the domains of $A^\sharp$, $B^\sharp$ and $N^\sharp$, where $N=BA$. Now consider the two sets \begin{equation} \mathcal{F}_{\Psi}=\{\Psi_{ n}, \,n\geq0\} \ \ \mbox{and} \ \ 
 \mathcal{F}_{\varphi}=\{\varphi_{ n}, \,n\geq0\}.
 \end{equation}
It is now simple to deduce the following lowering and raising relations:
\be
\left\{
    \begin{array}{ll}
B\,\varphi_n=\sqrt{n+1}\varphi_{n+1}, \qquad\qquad\quad\,\, n\geq 0,\\
A\,\varphi_0=0,\quad A\varphi_n=\sqrt{n}\,\varphi_{n-1}, \qquad\,\, n\geq 1,\\
A^\dagger\Psi_n=\sqrt{n+1}\Psi_{n+1}, \qquad\qquad\quad\, n\geq 0,\\
B^\dagger\Psi_0=0,\quad B^\dagger\Psi_n=\sqrt{n}\,\Psi_{n-1}, \qquad n\geq 1,\\
       \end{array}
        \right.
\label{A3}\en as well as the eigenvalue equations \begin{equation}N\varphi_n=n\varphi_n  \ \ \mbox{and} \ \  N^\dagger\Psi_n=n\Psi_n,
\end{equation} which are valid for all $n\geq0$. In particular, as a consequence
of these  last two equations,  choosing the normalization of $\varphi_0$ and $\Psi_0$ in such a way that $\left<\varphi_0,\Psi_0\right>=1$ is satisfied, we deduce that
\be \left<\varphi_n,\Psi_m\right>=\delta_{n,m}, \label{A4}\en
 for all $n, m\geq0$. Hence $\mathcal{F}_{\Psi}$ and $\mathcal{F}_{\varphi}$ are biorthonormal.

 Our third assumption is the following:

\vspace{2mm}

{\bf Assumption $\D$-pb 3.}  $\mathcal{F}_{\varphi}$ is a basis for $\mathsf{H}$.

\vspace{1mm}

This is equivalent to requiring that $\mathcal{F}_{\Psi}$ is a basis for $\mathsf{H}$ as well, \cite{chri}. However, several  physical models suggest to adopt the following weaker version of this assumption, \cite{bag2022book}:

\vspace{2mm}

{\bf Assumption $\mathcal{D}$-pbw 3.}  For some subspace $\mathcal{G}$ dense in $\mathsf{H}$, $\mathcal{F}_{\varphi}$ and $\mathcal{F}_{\Psi}$ are $\mathcal{G}$-quasi bases.

\vspace{2mm}
This means that, for all $f$ and $g$ in $\mathcal{G}$,
\be
\left<f,g\right>=\sum_{n\geq0}\left<f,\varphi_n\right>\left<\Psi_n,g\right>=\sum_{n\geq0}\left<f,\Psi_n\right>\left<\varphi_n,g\right>,
\label{A4b}
\en
which can be seen as a weak form of the resolution of the identity, restricted to $\D$. The role, and the necessity, of these sets is discussed in details in \cite{bag2022book}. For certain discussions which will be useful later on, it is useful to denote by 
\begin{equation}B(\mathsf{H})  = \{ T : \mathsf{H} \to \mathsf{H} \mid {\| T \|}_{\infty} =\sup_{x \in \mathsf{H}} \| T(x) \| < \infty\},
\end{equation} the Banach space of all bounded operators on $\mathsf{H}$.

Since $A$ and $B$ are unbounded, we can still set up a general algebraic settings as we did in \cite{bagrus2018}, using the notion of $*$-$algebra$, see \cite{aitbook, bag2022book, hofmor}. In particular we can construct what is called a {\em quasi $*$-algebras}, (see \cite{aitbook, bag2022book}). This is what we will do now, since it seems more useful for our particular problem. The notion of $O^*$-algebra is recalled below:

\begin{definition}[Domain of Closable Operators, see \cite{aitbook}]\label{o*}Let $\mathsf{H}$ be a separable Hilbert space and $N_0$ an
	unbounded, densely defined, self-adjoint operator. Let $D(N_0^k)$ be
	the domain of the operator $N_0^k$, $k \ge 0$, and $\mathcal{D}$ the domain of
	all the powers of $N_0$, that is,  \begin{equation}\label{dset}\mathcal{D} = D^\infty(N_0) = \bigcap_{k\geq 0}
	D(N_0^k). \end{equation} 
\end{definition}

The set $\mathcal{D}$ of Definition \ref{o*} is dense in $\mathsf{H}$ and it allows us to introduce another important concept:

\begin{definition}[Notion of  $O^*$-algebra, see \cite{aitbook}] \label{o**}Referring to \eqref{dset}, we define	$\mathcal{L}^\dagger(\mathcal{D})$ as the $*$-algebra of all the \textit{  closable operators} 	 on $\mathcal{D}$ which, together with their adjoints, map $\mathcal{D}$ into
	itself. Here the adjoint of $X\in\mathcal{L}^\dagger(\mathcal{D})$ is
	$X^\dagger=X^*_{| \mathcal{D}}$. We refer to  $\mathcal{L}^\dagger(\mathcal{D})$ as   $O^*$-$algebra$.
\end{definition}

Note that in Definitions \ref{o*} and \ref{o**} one needs to introduce a map
which, given an element $X\in\mathcal{L}^\dagger(\mathcal{D})$, produces another
element $X^\dagger \in \mathcal{L}^\dagger(\mathcal{D})$. The most natural choice,
which is clearly $X^\dagger\equiv X^*$, is  only compatible with
$\mathcal{L}^\dagger(\mathcal{D})=B(\mathsf{H})$, i.e. with $N_0$ bounded, which is not what
we want. Recalling that $D(X^*)\supseteq\mathcal{D}$, it is clear that
$X^*_{| \mathcal{D}}$ is well defined. Furthermore one can prove that
$\dagger$ has the properties of an involution and maps
$\mathcal{L}^\dagger (\mathcal{D})$ into itself. We shall also recall that
in $\mathcal{D}$ the topology is defined by the following $N_0$-depending
seminorms: \begin{equation}\phi \in \mathcal{D} \rightarrow \|\phi\|_n\equiv \|N_0^n\phi\|,
\end{equation}
where $n \ge 0$, and  the topology $\tau_0$ in $\mathcal{L}^\dagger(\mathcal{D})$ is introduced by the seminorms
\begin{equation} X\in \mathcal{L}^\dagger(\mathcal{D}) \rightarrow \|X\|^{f,k} \equiv
\max\left\{\|f(N_0)XN_0^k\|,\|N_0^kXf(N_0)\|\right\},
\end{equation} where
$k \ge 0$ and   $f \in \mathcal{S}(\mathbb{R})$. In this situation we say that
$\mathcal{L}^\dagger(\mathcal{D})[\tau_0]$ is a \textit{ complete $*$-algebra}.

%What we want to do now is to extend the ordinary construction well known for the CCR  to our case. For bosons we know that a vacuum $e_0$ does exist which is annihilated by $c$, $c\,e_0=0$, and which belongs to the domain of all the powers of $c^\dagger$. 

As a consequence, if $x,y\in \mathcal{L}^\dagger(\mathcal{D})$, we can multiply the two elements and the results, $xy$ and $yx$, both belong to $\mathcal{L}^\dagger(\mathcal{D})$. This is what is relevant for us: a suitable structure in which we have the possibility of introducing Lie brackets for objects (linear operators, in our case), which are not everywhere defined.
In fact, for each pair $x,y\in\mathcal{L}^\dagger(\mathcal{D})$, we can define a map $[ \ , \ ]$ as follows:
\begin{equation}\label{B1}
	[x,y]=xy-yx.
\end{equation}
It is clear that $[x,y]\in\mathcal{L}^\dagger(\mathcal{D})$, that it is bilinear, that $[x,x]=0$ for all $x\in\mathcal{L}^\dagger(\mathcal{D})$, and that it satisfies the Jacobi identity
\begin{equation}
[x,[y,z]]+[y,[z,x]]+[z,[x,y]]=0,
\end{equation}
for all $x,y,z\in\mathcal{L}^\dagger(\mathcal{D})$. Therefore, $[ \ , \ ]$ is a Lie bracket defined on $\mathcal{L}^\dagger(\mathcal{D})$.

\begin{proposition}[Lie Algebra of Pseudobosonic Operators, see \cite{bag2022book}]\label{liealgebraofldag} Given two $\mathcal{D}$-pseudobosonic operators $A$ and $B$, it is always possible to construct a complete $*$-algebra $\mathcal{L}^\dagger(\mathcal{D})$ which possesses also the structure of Lie algebra with respect to the Lie bracket  \eqref{B1}.
\end{proposition}

It appears to be natural to ask whether we can find connections with the deformation theory of Lie algebras of Gerstenhaber and others in \cite{gerstenhaber1, gerstenhaber2, gerstenhaber3, gerstenhaber4,  grunewald1, grunewald2} or not. The answer is positive and somehow is the new perspective which we want to offer in the present contribution. 

Before to get into the details, we shall note that  Proposition \ref{liealgebraofldag} allowed us to find  (see for instance \cite{bagrus2018, bagrus2019, bagrus2020}) that a series of dynamical systems possessing PT-symmetries were realizable via finite dimensional nilpotent complex Lie algebras.

\begin{proposition}[Construction of Lie Algebras by Bosons, see \cite{bagrus2018}, Proposition 4.3]\
\label{laconstruction}
\begin{itemize}

\item[(1).] We may introduce the self-adjoint position and momentum operators, $x$ and $p$, satisfying the commutation rule $[x,p]=i\mathbb{I}$;

\item[(2).] Using $x$ and $p$, we may define $N_0$ (essentially) as the Hamiltonian of the quantum harmonic oscillator: $N_0=p^2+x^2$, as well as their linear combinations operators $C=\frac{1}{\sqrt{2}}(x+ip)$ and $C^\dagger=\frac{1}{\sqrt{2}}(x-ip)$;

\item[(3).] We construct the algebra $\mathcal{L}^\dagger(\mathcal{D})$ as in Definition \ref{o*};

\item[(4).] We observe that $C, C^\dagger\in\mathcal{L}^\dagger(\mathcal{D})$, together with $N_0=2C^\dagger C+\mathbb{I}$;

\item[(5).] We note that $C$ and $C^\dagger$ are unbounded, as well as $N_0$. Therefore, the product of $N_0$ and $C$ makes no sense in $B(\mathsf{H})$,  but it is perfectly defined in $\mathcal{L}^\dagger(\mathcal{D})$;

\item[(6).]We note that the topology on $\mathcal{D}$ in  Definition \ref{o*} coincides with the standard topology $\tau_\mathcal{S}$ on $\mathcal{S}(\mathbb R)$  in  \eqref{sr}.

\end{itemize}
\end{proposition}

Proposition \ref{laconstruction} is the starting point to define an algebraic structure via pseudobosons.  

\begin{remark}\label{similiaritycondition} We report some of the main results in \cite{bagrus2018}. 
Let  $T$ be an invertible positive operator on $\mathsf{H}$ such that $T,{(T^{-1})}^\dag \in\mathcal{L}^\dagger(\mathcal{D})$. We may introduce two operators by similarity:
\begin{equation}\label{B2}
A=TCT^{-1} \ \ \mbox{and} \ \  B=TC^\dagger T^{-1}.
\end{equation}
It is clear that, because of our assumption on $T$, we have that $A,B\in \mathcal{L}^\dagger(\mathcal{D})$. The adjoints of $A$ and $B$ can be easily computed,
\begin{equation}
A^\dagger=(T^{-1})^\dagger C^\dagger T^\dagger \ \mbox{and} \  B^\dagger=(T^{-1})^\dagger C T^\dagger,
\end{equation}
and they both belong to $\mathcal{L}^\dagger(\mathcal{D})$. Moreover $A$ and $B$ satisfy Definition \ref{carrules}: they both map, with their adjoints, $\mathcal{D}$ into $\mathcal{D}$ and, taken any $f\in\mathcal{D}$ we have $abf-baf=f$. It is also clear that Assumptions $\mathcal{D}$-pb 1 and $\mathcal{D}$-pb 2 are both satisfied. In fact, the function \begin{equation}e_0(x)=\frac{1}{\pi^{1/4}}\,e^{-x^2/2}\end{equation} belongs to $\mathcal{S}(\mathbb R)$ and is annihilated by $C$. Since $T$ maps $\mathcal{S}(\mathbb R)$ into itself, and since $T$ is invertible,  the nonzero function $\varphi_0(x)=Te_0(x)$ belongs to $\mathcal{S}(\mathbb R)$ and is annihilated by $A$. Analogously the nonzero function $\Psi_0(x)=(T^\dagger)^{-1}e_0(x)$  belongs to $\mathcal{S}(\mathbb R)$ and is annihilated by $B^\dagger$. \end{remark}

It was shown in a series of recent contributions which is possible to realize any finite dimensional complex nilpotent Lie algebra as appropriate central extension of pseudobosonic operators, see \cite[Theorem 3.3]{bagrus2020}. Actually it is possible to construct  solvable nonnilpotent finite dimensional complex Lie algebras \cite[Theorem 3.5]{bagrus2020}. In particular, we report three concrete mathematical models, which motivated the investigations in \cite{bagrus2018, bagrus2019, bagrus2020}. These are connected with generalized harmonic oscillators and with the Swanson model \cite{swanson}.

\begin{theorem}[See \cite{bagrus2018}, Theorems 5.1, 5.2, 5.3]\label{main123}\ \begin{itemize}
\item[(i).]The pseudobosonic operators connected with the shifted harmonic oscillator admit a Lie algebra structure $\mathfrak{a}_{sh}$ which
is isomorphic to  the finite dimensional nilpotent complex Lie algebra $\mathfrak{n}_{5,2} = \mathfrak{h}(1) \oplus \mathfrak{i} \oplus \mathfrak{i}$ of dimension 5. 
\item[(ii).] The pseudobosonic operators connected with the deformed position and momentum operators admit a Lie algebra structure $\mathfrak{b}$, which is isomorphic to  $\mathfrak{a}_{sh}$.
\item[(iii).] The operators for the Swanson model    admit a Lie algebra structure $\mathfrak{s}$ which is a nonabelian  nilpotent Lie algebra of nilpotency class 2 with $[\mathfrak{s},\mathfrak{s}]=Z(\mathfrak{s})$ of dimension 1 and $\mathfrak{s}/Z(\mathfrak{s})$ abelian of dimension $4$. Moreover $\mathfrak{s}$ is isomorphic neither to $\mathfrak{h}(2)$ nor to $\mathfrak{a}_{sh}$.
\end{itemize}
\end{theorem}

\section{Main results concerning pseudobosonic operators}\label{sec4}

We are ready to prove our main results in the present section.

\begin{theorem}[First Main Theorem]\label{main1}
Up to a linear deformation $\mu_t$ of a Lie bracket $\mu$ of a complex finite dimensional nilpotent Lie algebra $\mathfrak{g}$, there exists  a unique realization via pseudobosonic operators of $\mathfrak{g}$ when $\mathrm{dim} \ \mathfrak{g} \le 5$. In particular, there exists a unique realization via pseudobosonic operators for the Lie algebras in {\rm (i), (ii) } and {\rm (iii)} of Theorem \ref{main123}.
\end{theorem}

\begin{proof}
If  $\mathrm{dim} \ \mathfrak{g} \le 5$, then Conjecture \ref{goh} is true by Theorem \ref{niceconfirmation}. 
Note also that a complete classification of $\mathfrak{g}$ is given in terms of generators and relations by Proposition \ref{classification}. It is then sufficient to provide a construction of $\mathfrak{g}$ and this will be unique (in the sense of the Grunewald-O'Halloran Conjecture). If $\mathrm{dim} \ \mathfrak{g} \in \{1,2\}$, there is nothing to show since there are only trivial cases. If $\mathrm{dim} \ \mathfrak{g} =3$, then we may construct $\mathfrak{h}(1)$ with pseudobosonic operators, as done in \cite{bagrus2018, bagrus2019, bagrus2020}, then Conjecture \ref{goh} implies that all the other complex 3-dimensional nilpotent Lie algebras would be obtained by $\mathfrak{h}(1)$ by a deformation (we have only $\mathfrak{n}_{3,1}$ in this situation). The result is true in dimension three. Then we go ahead and consider $\mathrm{dim} \ \mathfrak{g} = 4$. Looking at Proposition \ref{classification}, we have three types of Lie algebras. We construct one of these with pseudobosonic operators, as done in \cite[Theorem 3.3]{bagrus2020}, and the result follows for the same reason. Note that an alternative way to the application of \cite[Theorem 3.3]{bagrus2020} is to consider $\mathfrak{n}_{4,2}=\mathfrak{h}(1) \oplus \mathfrak{i}$ in Proposition \ref{classification} (2) and observe that $\mathfrak{i}$ a monodimensional abelian algebra, which commutes with  the elements of $\mathfrak{h}(1)$. Finally, the same logic applies with $\mathrm{dim} \ \mathfrak{g} = 5$. For sake of completeness, we illustrate  the construction. The idea is to use the construction for the shifted harmonic oscillator in (i) of Theorem \ref{main123}. We rewrite it here, overlapping what is known from the proof of (i) of Theorem \ref{main123} in \cite{bagrus2018}.

Firstly, we define
\begin{equation}\label{C1}
A=C-\alpha\mathbb{I} \ \ \mbox{and} \ \  B=C^\dagger-\overline{\beta}\mathbb{I},
\end{equation}
where $\alpha$ and $\beta$ are different complex quantities such that $\alpha\neq\beta$. Hence $B \neq A^\dagger$, clearly. To check that $A$ and $B$ belong to $\mathcal{L}^\dagger(\mathcal{D})$, we note that $C,C^\dagger$ and $\mathbb{I}$ belong to $\mathcal{L}^\dagger(\mathcal{D})$, and since this is a linear vector space, we get $A,B\in \mathcal{L}^\dagger(\mathcal{D})$.  This argument applies in analogy to $A^\dagger$ and $B^\dagger$.

Secondly, we consider  $v_1=A$, $v_2=B$, $v_3=B^\dagger$, $v_4=A^\dagger$, $v=\mathbb{I}$ and  the Lie algebra.
\begin{equation}\label{C2}
\mathfrak{a}_{sh} = \langle v_1, v_2, v_3, v_4, v \  |  \  [v_1,v_2]=[v_3,v_4]=[v_1,v_4]=[v_3,v_2]=v,
\end{equation}
$$[v_1,v_3]=[v_2,v_4]=[v,v_j]=0 \quad \forall j=1,2,3,4\rangle.$$
%Note that the Lie bracket, involved in the above presentation of $\mathfrak{a}_{sh}$, are exactly the rules which we have introduced, and commented, in \eqref{B1}.

Finally, we note that this is a finite dimensional nilpotent Lie algebra of dimension 5. In fact $\mathfrak{a}_{sh}=\mathfrak{n}_{5,2}$ according to Proposition \ref{classification}. The result follows.
\end{proof}

In a moment we will see that Theorem \ref{main1} holds even if $\mathrm{dim} \ \mathfrak{g} \ge 6$, but here one needs  some extra assumption. We  indicate the relevance of Theorem \ref{main1} with a brief consideration here. If we look at (3) of Proposition \ref{classification}, there are 9 different types of complex 5-dimensional nilpotent Lie algebras. According to Theorem \ref{main1}, we may construct one of them via pseudobosonic operators, hence Conjecture \ref{gammasequences} is true, and so any other of these 9 complex nilpotent Lie algebras is obtained by a deformation. From the point of view of mathematical physics, this observation is relevant, since it means that we may consider a single type of quantum dynamical system and all the others would be obtained by this up to appropriate deformations.

\begin{theorem}[Second Main Theorem]\label{main2}
Up to a linear deformation $\mu_t$ of a Lie bracket $\mu$, any finite dimensional complex nilpotent Lie algebra of dimension $\ge 6$ possessing  a nontrivial semisimple derivation may be obtained as deformation of another finite dimensional nilpotent Lie algebra which is  realized by pseudobosonic operators. 
\end{theorem}

\begin{proof}Consider a complex finite dimensional nilpotent Lie algebra $\mathfrak{g}$ of  dimension $\mathrm{dim} \ \mathfrak{g}$ bigger than six possessing nontrivial semisimple derivations. In \cite[Theorem 3.3]{bagrus2020} it has bee proved that all finite dimensional complex nilpotent Lie algebras may be realized
by extensions of pseudobosonic operators. This means that there exists a realization of $\mathfrak{g}$ with pseudobosonic  operators (and in fact this can be obtained computationally with the method of Skjelbred and Sund, see Remark \ref{crucialobservation} for details). Now we apply Conjecture \ref{goh}, which is true in the present situation by \cite[Theorem 1]{herrera1}. The result follows.
\end{proof}

 Theorem \ref{main2} shows an important connection between the theory of the deformations of Lie algebras, elaborated by Gerstenhaber and others \cite{burde1, burde2, gerstenhaber1, gerstenhaber2, gerstenhaber3, grunewald1, grunewald2, rem1, seeley1, seeley2}, with the theory of ladder operators, and specifically with the quantum dynamical systems which are subject to a formalization in terms of pseudobosonic operators, see \cite{bag2022book, bagcohsta, bagrus2018, bagrus2019, bagrus2020, bender1}. To the best of our knowledge, the first time that this connection has been noted is here.

\section{Quons and deformations which do not preserve the Jacobi identity}\label{sec5}

We begin to observe that Remark \ref{jacobiidentity} and Proposition \ref{preservationproperty1} illustrate that we have a good enough deformation theory for complex finite dimensional Lie algebras in the sense that Jacobi identity is preserved by appropriate deformations, see Definition \ref{lindef}. We may consider certain associative algebras which do not possess a Lie bracket in the usual sense, so there is no hope to look for Jacobi identity, but  generalized Lie brackets (as indicated in Appendix II of Section \ref{qdeformedappendix}) and generalized Jacobi identities may appear, see Proposition  \ref{generalizedjacobi3}. 

\begin{definition}[Notion of q-Deformed Heisenberg Algebra, see \cite{sergey2}] \label{qh} The $q$-deformed Heisenberg algebra $\mathcal{H}(q)$ is an associative algebra over the complex field with unit (denoted by $\mathbb{I}$) and defined by a presentation with generators $A, B$ and relations \begin{equation}\label{qmutator} {[A,B]}_q=\mathbb{I},\end{equation} where ${[A,B]}_q=AB-qBA$ is called $q$-\textit{mutator} of $A$ and $B$ and  $q \in [-1,1]$ is a parameter. \end{definition}

Note that \eqref{qmutator} in Definition \ref{qh} brings to an  algebraic formalism of the so-called $q$-\textit{deformed canonical commutation relation}, where the CCR in \eqref{carrules} are obtained in the specific case that $q=1$. In quantum mechanics, it is a common practice to think at Definition \ref{qh} with  $A$ and $B$ which are  mutually adjoint linear operators (with $\mathbb{I}$ as the identity map) on an infinite dimensional separable Hilbert space.  Then we briefly review here few basic facts on \textit{quons} from \cite{bagquons, fiv, moh,gre}, since these particles of the matter turn out to be more general than pseudobosons and they extend both bosons and  fermions. Quons are
defined essentially by their \textit{q-mutation relation} \be [C,C^\dagger]_q:=C C^\dagger -q C^\dagger C
=\mathbb{I}, \qquad q\in [-1,1], \label{21} \en
between the creation and the annihilation operators $C^\dagger$
and $C$, which reduces to the CCR for $q=1$ and to the CAR for $q=-1$. Note that the condition $C^2=0$ is not required here a priori. For  $q$ in
the interval $]-1,1[$, we have that \eqref{21} describes particles which are neither
bosons nor fermions.

%{\color{blue} It is interesting to notice that, using (\ref{21}), we can slightly extend what we discussed before for ordinary commutators, and for their deformations. We propose here the following approach:	first of all, we observe that, introducing	\be\label{2a}[x,y]_q=xy-qyx,\enfor all $x,y\in \mathfrak{A}$, our relevant algebra and $q\in[-1,1]\setminus\{0\}=:I$, the following equality is satisfied:\be \label{2b}[x,y]_q=-q[y,x]_{1/q}, \qquad q\in I.\enIn particular, if $q=1$, this equality reduces to say that $[.,.]_1$ is antisymmetric, which is obviously in agreement with the analogous well known property of all the Lie brackets, and of the commutators in particular.Driven by, we introduce a family $\mathcal{F}$ of maps $\Phi_q(.,.)$ from $\mathfrak{A} \times \mathfrak{A}$ to $\mathfrak{A}$, depending on a parameter $q \in I$ having the properties that\be\label{2c} \Phi_q(x,y)=-q\,\Phi_{1/q}(y,x),\en$\forall x,y\in \mathfrak{A}$. It is easy to check that, if $s\in \mathfrak{A}$ is an invertible element of $\mathfrak{A}$, with $s^{-1}\in \mathfrak{A}$, then the new map$$\Phi_q^{(s)}(x,y)=s\Phi_q(x,y)s^{-1}$$is again an element of $\mathcal{F}$. The check is immediate and will not be given here. Hence it is clear that $\mathcal{F}$ contains all the forms which are similar to the q-mutator in (\ref{2b}). What is not clear is if the opposite is also true: is any element of $\mathcal{F}$ similar to some q-mutator? Our idea is that this is not really so, but we don't have, for the moment, a counterexample. }

In \cite{moh} it is proved that the eigenstates of $N_0=C^\dagger\,C$ are analogous to the bosonic ones, except that for the normalization. A simple concrete realization of  \eqref{21} can be deduced as follows:

\begin{example}[Mohapatra's Construction, see \cite{moh}] \label{mohex}Let $\mathcal{F}_e=\{e_k \mid k=0,1,2,\ldots\}$ be the canonical orthonormal  basis of the Hilbert space $\mathsf{H}=l^2(\mathbb N_0)$, with all zero entries except in the $(k+1)$-th position, which is equal to one. In this situation we involve the Kronecker symbol and get a scalar product of the following form \begin{equation} \langle e_k,e_m\rangle=\delta_{k,m}.\end{equation}  
If we take
\begin{equation}
C=\left(
    \begin{array}{ccccccc}
      0 & \beta_0 & 0 & 0 & 0 & 0 & \cdots \\
      0 & 0 & \beta_1 & 0 & 0 & 0 & \cdots \\
      0 & 0 & 0 & \beta_2 & 0 & 0 & \cdots \\
      0 & 0 & 0 & 0 & \beta_3 & 0 & \cdots \\
      0 & 0 & 0 & 0 & 0 & \beta_4  & \cdots \\
      \cdots & \cdots & \cdots & \cdots & \cdots & \cdots & \cdots \\
      \cdots & \cdots & \cdots & \cdots & \cdots & \cdots & \cdots \\
    \end{array}
  \right),
\label{21bis}\end{equation}
then \eqref{21} is satisfied if \begin{equation} \beta_0^2=1 \ \ \mathrm{and} \ \  \beta_n^2=1+q\beta_{n-1}^2, \ \ \  \forall n\geq1.\end{equation} and we have that
\be
\beta_n^2=\left\{
    \begin{array}{ll}
n+1, \  \quad \mbox{if} \   \ q=1,\\
\\
\frac{1-q^{n+1}}{1-q}, \ \ \mbox{if} \ \ q\neq 1.
       \end{array}
        \right.
\en

We always consider $\beta_n>0$ for all $n\geq0$. The above form of $C$ shows that  $C\,e_0=0$, and $C^\dagger$ behaves as a raising operator. From (\ref{21bis}) we deduce
\be
e_{n+1}=\frac{1}{\beta_n}\,C^\dagger e_n=\frac{1}{\beta_n!}(C^\dagger)^{n+1}\,e_0,
\label{22}\en
for all $n\geq0$. Here we have introduced the $q$-\textit{factorial}: \begin{equation}{[n]}_q=\beta_n!=\beta_n\beta_{n-1}\cdots\beta_2\beta_1.\end{equation}  Note that the symbol   ${[n]}_q)$ is often used in q-special combinatorics, see \cite{eremel} or  Proposition \ref{generealizedjacobi1} later on, and note also that  $\beta_n^2={[n+1]}_q$. Of course, from \eqref{22} it follows that  \begin{equation}C^\dagger e_n=\beta_ne_{n+1}.\end{equation} Using \eqref{21bis} it is also easy to check that $C$ acts as a lowering operator on $\mathcal{F}_e$: \begin{equation}C\,e_m=\beta_{m-1}e_{m-1} \ \ \ \ \ \ \ \forall m \ge 0,\end{equation} where it is also useful to introduce $\beta_{-1}=0$, in order to ensure  $C\,e_0=0$.

 Then we have
\be
N_0e_m=\beta_{m-1}^2e_m, \ \ \ \ \ \ \ \forall m \ge 0.
\label{23}\en
The operator $N$ is formally defined in \cite{moh} as follows \begin{equation}N=\frac{1}{\log(q)}\,\log(\mathbb{I}-N_0(1-q)) \ \ \ \ \ \mbox{for}  \ q>0,\end{equation} and satisfies the eigenvalue equation \begin{equation}Ne_m=me_m, \ \ \ \ \ \forall m\ \ge 0.\end{equation}
\end{example}

It  is good to note that there are different ways to represent the operator $C$ in \eqref{21bis}, that is, there are alternative constructions to those which have been proposed in Example \ref{mohex}.

\begin{example}[Eremin--Meldianov's Construction, see \cite{eremel}]\label{mohexbis}We may begin with  $C$ and $C^\dagger$ in the Hilbert space $\mathsf{H}=L^2(\mathbb R)$:
\be
C=\frac{e^{-2i\alpha x}-e^{i\alpha \frac{d}{dx}}e^{-i\alpha x}}{-i\sqrt{1-e^{-2\alpha^2}}}, \qquad\qquad C^\dagger=\frac{e^{2i\alpha x}-e^{i\alpha x}e^{i\alpha \frac{d}{dx}}}{i\sqrt{1-e^{-2\alpha^2}}},
\label{23bis}\en
where $\alpha=\sqrt{-\frac{\log(q)}{2}}$ or, which is the same, $q=e^{-2\alpha^2}$. However, since $\alpha$ is assumed to belong to the set  $[0,\infty)$,  $q$ ranges in the interval $]0,1]$. Then, the representation in (\ref{23bis}) only works in this region. Neverthless, one can check that  relations hold $[C,C^\dagger]_q:=C C^\dagger -q C^\dagger C
=\mathbb{I}$.
\end{example}

The considerations  in Section \ref{sec3} can be somehow {\em merged}, giving rise to what we call $\D$-pseudo quons, ($\D$-PQs). The procedure is very close to that for ordinary bosons, but requires some careful observations. We always begin with  $\mathsf{H}$  Hilbert space with scalar product $\langle \ , \ \rangle$ and related norm $\| \ \|$. Then we consider $A$ and $B$  operators
on $\mathsf{H}$ with domains $D(A)$ and $D(B)$ respectively. We pass to consider $A^\dagger$ and $B^\dagger$ which are their adjoints, and again we consider $\D$ dense subspace of $\mathsf{H}$
such that $A^\sharp\D\subseteq\D$ and $B^\sharp\D\subseteq\D$. The framework of functional analysis is exactly the same that we have seen in Section \ref{sec3} when we discussed the pseudobosonic operators. In particular, note that $\D\subseteq D(A^\sharp)$
and $\D\subseteq D(B^\sharp)$.

\begin{definition}[q-mutators for Pseudoquonic Operators, see \cite{bagquons}]\label{qcarrules}
The operators $(A,B)$ are $\D$-pseudoquonic  if for all $f\in\D$ we have
\be
[A,B]_qf:=A\,B\,f-q \,B\,A\,f=f,
\label{31}\en
for some real number $q \in [-1,1]$. In particular, for $q=1$ we find Definition \ref{carrules}.
\end{definition}

In principle,  we do not need to restrict $q \in [-1,1]$, but we put this assumption in order to make a perfect analogy with dynamical systems which looks like pseudobosonic operators. We then formalize the next steps as done in Section \ref{sec3}.

\vspace{2mm}

{\bf Assumption $\D$-pq 1.}  There exists a non-zero $\varphi_{ 0}\in\D$ such that $A \,\varphi_{ 0}=0$.

\vspace{1mm}

{\bf Assumption $\D$-pq 2.}  There exists a non-zero $\Psi_{ 0}\in\D$ such that $B^\dagger\,\Psi_{ 0}=0$.

\vspace{2mm}

If $(A,B)$ satisfy Definition \ref{qcarrules}, then \be \varphi_n:=\frac{1}{{\beta_{n-1}!}}\,B^n\varphi_0=\frac{1}{\beta_{n-1}}\,B\,\varphi_{n-1} \quad \mbox{and} \quad \Psi_n:=\frac{1}{{\beta_{n-1}!}}\,{A^\dagger}^n\Psi_0=\frac{1}{{\beta_{n-1}}}\,{A^\dagger}\,\Psi_{n-1}, \quad   \forall n\geq1,\en can be defined and  belong to $\D$, but also to the domains of $A^\sharp$, $B^\sharp$ and $N^\sharp$, where $N=BA$. Note that $\beta_n$ are those in Example \ref{mohex}. We further define $\mathcal{F}_\Psi=\{\Psi_{ n}, \,n\geq0\}$ and
$\mathcal{F}_\varphi=\{\varphi_{ n}, \,n\geq0\}$.

Then the following lowering and raising relations  replace those in \eqref{A3}:
\be
\left\{
    \begin{array}{ll}
B\,\varphi_n=\beta_n\varphi_{n+1}, \qquad\qquad\qquad\qquad\qquad\quad\,\, n\geq 0,\\
A\,\varphi_0=0,\quad \varphi_n=\beta_{n-1}\,\varphi_{n-1}, \qquad\qquad\quad\,\, n\geq 1,\\
A^\dagger\Psi_n=\beta_n\Psi_{n+1}, \qquad\qquad\qquad\qquad\quad\quad\quad n\geq 0,\\
B^\dagger\Psi_0=0,\quad B^\dagger\Psi_n=\beta_{n-1}\,\Psi_{n-1}, \qquad\quad\qquad n\geq 1.\\
       \end{array}
        \right.
\label{32}\en
In a more compact way we can rewrite the second and the fourth of these equalities as follows: $a\varphi_n=\beta_{n-1}\,\varphi_{n-1}$ and $b^\dagger\Psi_n=\beta_{n-1}\,\Psi_{n-1}$, with the agreement that $\Psi_{-1}=\varphi_{-1}$ are the zero vectors. Recall also that $\beta_{-1}=0$. Then, from (\ref{32}), using the stability of the set $\D$ under the action of $b$ and $a^\dagger$, we deduce the eigenvalue equations
$N\varphi_n=\beta_{n-1}\varphi_n$ and  $N^\dagger\Psi_n=\beta_{n-1}\Psi_n$, $n\geq0$ which show that $N$ and $N^\dagger$ have the same eigenvalues. This suggests the existence of some intertwining operator between $N$ and $N^\dagger$. This will be discussed later. Furthermore,  choosing the normalization of $\varphi_0$ and $\Psi_0$ in such a way $\left<\varphi_0,\Psi_0\right>=1$, we conclude, as in Section \ref{sec3}, that
\be \left<\varphi_n,\Psi_m\right>=\delta_{n,m}, \label{33}\en
 for all $n, m\geq0$. Hence $\mathcal{F}_\Psi$ and $\mathcal{F}_\varphi$ are biorthonormal. Our third assumption is the following:

\vspace{2mm}

{\bf Assumption $\mathcal{D}$-pq 3.--}  $\mathcal{F}_\varphi$ is a basis for $\mathsf{H}$,

\vspace{1mm}

and this implies that $\mathcal{F}_\Psi$ is a basis for $\mathsf{H}$ as well, see \cite{chri}). Again the weaker version is motivated by the evidences which we have from pseudobosonic operators:

\vspace{2mm}

{\bf Assumption $\mathcal{D}$-pqw 3.--}  For some subspace $\mathcal{G}$ dense in $\mathsf{H}$, $\mathcal{F}_\varphi$ and $\mathcal{F}_\Psi$ are $\mathcal{G}$-quasi bases.

\vspace{2mm}

We refer to \cite{bagquons} for some concrete examples of pseudobosonic operators and for their realization in square summable spaces. 

\begin{example}[Generalized Version of Eremin-Meldianov's Construction, see \cite{bagquons}] \label{mohextris} Here we just give a specific example of $A$ and $B$ considered in \cite{bagquons}, offering another perspective for Example \ref{mohexbis} in the context of pseudoquonic operators. We consider
\be
A=\frac{e^{-2i\alpha x}-e^{i\alpha \frac{d}{dx}}e^{-i\alpha (x+\gamma)}}{-i\sqrt{1-e^{-2\alpha^2}}} \quad \mbox{and} \quad B=\frac{e^{2i\alpha x}-e^{i\alpha (x-\gamma)}e^{i\alpha \frac{d}{dx}}}{i\sqrt{1-e^{-2\alpha^2}}}.
\label{52}\en
Here $\alpha=\sqrt{-\frac{\log(q)}{2}}$ or, which is the same, $q=e^{-2\alpha^2}$. 
Of course, $(A,B)$ collapse to the pair $(C,C^\dagger)$ of Example \ref{mohexbis} if $\gamma=0$, but, if $\gamma\neq0$, then $B\neq A^\dagger$. 
\end{example}

At this stage of the construction, there are the new considerations to be done.

\begin{theorem}[Third Main Theorem]\label{main3} The
     $O^*$-algebra $\mathcal{L}^\dag(\mathcal{D})$ of the pseudoquonic operators on $\mathcal{D}$ as per Definition \ref{o**} possess the structure of $q$-deformed Heisenberg algebra. In particular, there exists some $q \in [-1,1]$ such that $\mathcal{L}^\dag(\mathcal{D}) \simeq \mathcal{H}(q)$.
\end{theorem}

\begin{proof}
We repeat the same procedure, which has been illustrated in Section \ref{sec3} in order to get Proposition \ref{liealgebraofldag} from Definition \ref{o**}, but this time we use the  commutator \eqref{31} instead of \eqref{A1}.  In this way we find that $\mathcal{L}^\dag(\mathcal{D})$ satisfies Definition \ref{qh} and the result follows.
\end{proof}

Despite the short proof of Theorem \ref{main3}, its consequences are relevant. First of all, Proposition \ref{liealgebraofldag} follows from Theorem \ref{main3} when $q=1$. Secondly,  Proposition \ref{laconstruction} follows from Theorem \ref{main3}.

\begin{corollary}[Construction of q-Deformed Heisenberg Algebras by Pseudoquons]\

\begin{itemize}

\item[(1).] Define $N_0$ as the Hamiltonian of the generalized quantum harmonic oscillator $N_0=C^\dag C$ in Example \ref{mohex} (or  Example \ref{mohexbis});

\item[(2).] Construct the algebra $\mathcal{L}^\dagger(\mathcal{D})$ as in Definition \ref{o*} from pseudoquonic operators $C$ and $C^\dagger$;

\item[(3).]Observe that $C, C^\dagger\in\mathcal{L}^\dagger(\mathcal{D})$, together with $N_0$;

\item[(4).] Apply Theorem \ref{main3} and find that  there exists some $q \in [-1,1]$ such that $\mathcal{L}^\dag(\mathcal{D}) \simeq \mathcal{H}(q)$.

\end{itemize}
\end{corollary}

\section{$q$-CCR for pseudoquonic operators and stability of Cuntz algebras}

Jorgensen and others \cite{qccr}  consider  operators in Definition \ref{qh} in the following form \begin{equation}\label{extra21}A_i A_j^\dag-q A_j^\dag A_i=\delta_{ij} \mathbb{I}  \ \  \ \ \  \mathrm{for} \ \ i,j \in \{1,2, ..., d\},\end{equation} where $d$ is a fixed positive integer, $q\in [-1,1]$ and $\delta_{ij}$ denotes the Kronecker delta. This is exactly the situation which is decribed in \eqref{21} and  the authors of \cite{qccr}   show the existence of a \textit{universal}  $C^*$-\textit{algebra} $\mathfrak{C}^q$ satisfying \eqref{extra21} with the property that every other $C^*$-algebra  of bounded operators satysfying the same  relations is a homomorphic image of $\mathfrak{C}^q$. Note that homomorphic images of $C^*$-algebras are surjective continuous maps which preserve both the algebraic structure and the $*$ operator. This result gives a condition of existence and uniqueness, but applies only for  small values of $q$ and in particular to $q=0$, where  $\mathfrak{C}^0$   is usually known as \textit{Cuntz algebra}.

Note that from \cite[Theorem 3.5.1]{baginbook} and \cite[Theorem 2.8.1]{bag2022book}  it is possible to produce pseudobosons and pseudofermions up to similarity with appropriate positive operators. This is possible also at the level of pseudoquons, as illustrated by the following result.Theorem 2.8.1

\begin{theorem}\label{deforminguptosimilarity}
   Given $i,j \in \{1,2, ..., d\}$ and $d$ positive integer, for every $q$-CCR of the form  $U_i U_j^\dag-q U_j^\dag U_i=\delta_{ij} \mathbb{I}$, there exists $q$-CCR of the form $V_i V_j-q V_j V_i=\delta_{ij} \mathbb{I}$ such that $U_i=T V_i T^{-1}$ and $U_j=T V_j^\dag T^{-1}$ for some positive invertible self-adjoint operator $T$. Viceversa, given  a positive invertible self-adjoint operator $T$ and a
$q$-CCR of the form $V_i V_j-q V_j V_i=\delta_{ij} \mathbb{I}$, we obtain a $q$-CCR of the form  $U_i U_j^\dag-q U_j^\dag U_i=\delta_{ij} \mathbb{I}$ up to similarity of $U_i$ and $V_i$ by $T$.
\end{theorem}

\begin{proof}     Assume we have a $q$-CCR of the form  $U_i U_j^\dag-q U_j^\dag U_i=\delta_{ij} \mathbb{I}$.    Then define $V_i=T U_i T^{-1}$ and $V_j=T U_j^\dag T^{-1}$ with the aforementioned properties for $T$. We find that
    \begin{equation}
        V_i V_j-q V_j V_i = T U_i T^{-1} T U_j^\dag T^{-1} - q T U_j^\dag T^{-1} T U_i T^{-1}
                =T U_i  U_j^\dag T^{-1} - q T U_j^\dag U_i T^{-1}
        \end{equation}
  \[      =T \left( \delta_{ij} \mathbb{I} + q U_j^\dag U_i \right) T^{-1} - q T U_j^\dag U_i T^{-1}        
       =T \delta_{ij} \mathbb{I} T^{-1} + T q U_j^\dag U_i T^{-1} - q T U_j^\dag U_i T^{-1}       
\]\[        =T \delta_{ij} \mathbb{I} T^{-1} + q T  U_j^\dag U_i T^{-1} - q T U_j^\dag U_i T^{-1}   =\delta_{ij} \mathbb{I}.\]
Conversely,     assume there exists $q$-CCR of the form $V_iV_j-q V_j V_i=\delta_{ij} \mathbb{I}$. Then it is possible to define operators $U_i, U_j$ and a positive invertible self-adjoint operator $T$ such that $V_i=T U_i T^{-1}$ and $V_j^\dag=T U_j T^{-1}$. A similar check implies $U_i U_j^\dag-q U_j^\dag U_i=\delta_{ij} \mathbb{I}$, so the result follows.
\end{proof}

We want to recall  the construction of  $\mathfrak{C}^q$ from \cite{qccr}, since it is pertinent to the discussion. The first step is to consider  the complex vector space $\mathcal{V}$  of the sesquilinear forms $\langle \cdot, \cdot \rangle$. Then one focuses on the $\mathbb{C}$-linear maps $r: \mathcal{V} \to \mathcal{A}$ into a $C^*$-algebra $\mathcal{A}$ with identity such that $    r(f)r(g)-qr(g)r(f)=\langle g, f \rangle \mathbb{I}$, for all $f,g \in \mathcal{V}.$ Any map $r$ with this property is called  \textit{representation of the} $q$-\textit{relations}. Now one notes separately what happens in the case of a on edimensional $C^*$-algebra of this kind, that is, when $\mathcal{V}$ is one dimensional.

\begin{proposition}[See \cite{qccr}, Proposition 1]
    Let $\mathcal{A}$ be a $C^*$-algebra with a unit, $q \in (-1,1)$, $c \in \mathbb{R}$, and $a \in \mathcal{A}$ a non-zero element satisfying
    $$aa^*-qa^*a=c \1.$$
    Then $c>0$, and either $aa^*=a^* a=\frac{c}{1-q}$, or the spectra of $aa^*$ and $a^*a$ are equal to the closure of the sequence $$c \frac{1-q^n}{1-q},$$
    where $n \in \mathbb{N}$, and $n \geq 0$ for $a^* a$ and $n \geq 1$ for $aa^*$.
\end{proposition}

\begin{corollary}
    Let $\mathcal{V}$ be a $C^*$-algebra with a unit, $q \in (-1,1)$, $c \in \mathbb{R}$, and $A,B \in \mathcal{A}$ some non-zero elements, some $T=T^\dagger$ with $T>0$, $a=TAT^{-1}$ and $a^*=T B T^{-1}$
    $$aa^*-qa^*a=c \1.$$
    Then $c>0$, and either $AB=BA=\frac{c}{1-q}$, or the spectra of $AB$ and $BA$ are equal to the closure of the sequence $$c \frac{1-q^n}{1-q},$$
    where $n \in \mathbb{N}$, and $n \geq 0$ for $BA$ and $n \geq 1$ for $AB$.
\end{corollary}

\begin{corollary}
    Suppose $r: \mathcal{V} \to \mathcal{A}$ satisfies the q-relations. Then the sesquilinear form $\langle \cdot, \cdot \rangle$ on $\mathcal{V}$ is positive semidefinite, and 
    \begin{align*}
        \|r(f)\| &= c(q) \langle f, f \rangle, \quad \text{ where }
        \\
        c(q)&=\left\{\begin{array}{ll}
        \frac{1}{\sqrt{1-q}} \ \ \ \ \text{ for } 0 \leq q < 1
        \\
        1 \quad \quad \quad \text{ for } -1 < q \leq 0
       \end{array}
        \right.
    \end{align*}
\end{corollary}

For this next Theorem we simply adapt the proof of Proposition 3 in \cite{qccr} to our Psuedoquouns. 

\begin{theorem}
    \begin{itemize}
        \item[(1)] For a given Hilbert space $\mathcal{V}$, and any $q \in (-1,1)$, there exists a $C^*$-algebra, denoted by $\mathfrak{C}^q(\mathcal{V})$ with a map $r: \mathcal{V} \to \mathfrak{C}^q(\mathcal{V})$ satisfying the $q$-relations with the following universal: for any map $\Tilde{r}: \mathcal{V} \to \mathcal{A}$ satisfying the $q$-relations there is a unique $*$-homomorphism $\Tilde{\beta}: \mathfrak{C}^q(\mathcal{V}) \to \mathcal{A}$ such that $\Tilde{\beta}\left(r(f)\right)=\Tilde{r}(f)$ for all $f \in \mathcal{V}$.
        \item[(2)] $\mathfrak{E}^q(\mathcal{V})$ is determined up to $C^*$-isomorphism. For any isometry $H: \mathcal{V}_1 \to \mathcal{V}_2$ between Hilbert spaces, there is a unique $*$-homomorphism $\mathfrak{C}^q(H): \mathfrak{C}^q(\mathcal{V}_1)\to \mathfrak{C}^q(\mathcal{V}_2)$ such that $\mathfrak{C}^q(H)(r(f))=r(Hf)$.
    \end{itemize}
\end{theorem}

\begin{proof}
    Let $\textbf{C}^q(\mathcal{V})$ denote the quotient if the $*$-algebra of non-commuting polynomials in the generators $r(f)$ by the ideal generated by the $q$-relations and the relations arising from the linearity of $r$. It is straightforward to check that $\textbf{C}^q(\mathcal{V})$ satisfies the analagous rules in of the claims above. If $\Tilde{r}$ satisfies the relations, the homomorphism $\Tilde{\beta}$ is defined simply by substituting the given $\Tilde{r}(f)$ into any polynomial in $\textbf{C}^q(\mathcal{V})$. In order to get a universal $C^*$-algebra, to satisfy our claims above, we have to define a $C^*$-seminorm $\| \cdot \|$ on $\textbf{C}^q(\mathcal{V})$ with the property that for any bounded representation of the relations the substitution homomorphism is continuous, and hence extends to the completion $\mathfrak{C}^q(\mathcal{V}_1)$ of $\textbf{C}^q(\mathcal{V)}$ with respect to that seminorm. This is to say that for any polynomial $X \in \textbf{C}^q(\mathcal{V})$, and any bounded realisation $\Tilde{r}$ we must have $\|\Tilde{\beta}(X)\| \leq \|X\|$. Hence we want to define 
    $$\|X\|=\sup\limits_{\Tilde{r}} \| \Tilde{\beta} \|,$$
    where the sup is over all bounded realisations of $\Tilde{r}$ of the relations, and $\Tilde{\beta}$ denotes the associated substitution homomorphism. There are three issues with this definition. The first is the technical point that the bounded realisations of the relations do not for a set, but a proper class. However since $\|\Tilde{\beta}(X)\|$ may be computed in the $C^*$-algebra generated by the $\Tilde{r}(f)$, we may restrict the supremum to realisations in $C^*$-algebras with at most a certain cardinal number of elements depending on the number of generators via the dimension of $\mathcal{V}$. We now pick a Hilbert space of sufficiently high dimension such that all the universal representations of all these algebras can be realised in it. It is then clear that the supremum can be restricted to the set of realisations of the $q$-relations in bounded operators on this big Hilbert space.
    The second is whether the supremum is possibly over the empty set. This case was settled in \cite{bz}, where a Fock representation for all $q \in (-1,1)$ was explicitly constructed (see \cite{qccr}). A \textit{Fock representation} is one where for each annhilation operator: $\partial_t \Omega=0$, where $\Omega$ is the vacuum vector.
    The third problem is that the supremum might turn out to be infinite for some $X$. Since the $X \in \textbf{C}^q(\mathcal{V})$ are polynomials in $\Tilde{r}(f)$ this can be ruled out by proving an upper bound on $\|\Tilde{r}(f)\|$, which is uniform with respect to all realisations $\Tilde{r}$ of the $q$-realisations. For the case at hand a bound of this kind was established in Corollary 2. Hence $\mathfrak{C}^q$ exists and has the universal property. The uniqueness follows from the universal property.
\end{proof}

\section{Conclusions and final remarks}\label{sec6}

On the basis of what we have seen in Section \ref{sec2}, it is clear that we may start to consider a 1-deformed Heisenberg algebra $\mathcal{H}(1)$ and this turns out to be a Lie algebra, hence the deformation theory of Lie algebras applies to $\mathcal{H}(1)$.

In case $q\neq1$ we also may consider a deformation theory for $\mathcal{H}(q)$, but it is not very clear how to get results of the form of Theorem \ref{niceconfirmation}. Let's be more specific. Looking at the works of Gerstenhaber, Silvestrov and others \cite{gerstenhaber1, gerstenhaber2, sergey1, sergey3}, it is possible to have a cohomology theory for $\mathcal{H}(q)$. In particular, one can introduce (after many technical preliminaries) a notion of Schur multiplier for $\mathcal{H}(q)$ which generalizes that of Schur multiplier of  Lie algebras in Remark \ref{derivation2}. When we do this, two problems come to our attention:

\begin{problem}\label{nicetostudy1}
Classification of nilpotent $q$-deformed Heisenberg algebras and  develop a method of algorithmic nature (as that in Remark \ref{derivation3}) to construct them.
\end{problem}

For nilpotent $q$-deformed Heisenberg algebras, we need to introduce as usual a notion of center and go ahead, iterating this notion, see \cite{sergey1, sergey2, sergey3} for more details. Only when we have more elements to go ahead in Problem \ref{nicetostudy1}, we can get to a significant answer to the following:

\begin{problem}\label{nicetostudy2}
Develop a deformation theory for $\mathcal{H}(q)$ in the sense of Gerstenhaber and Grunewald.
\end{problem}

Note that for $q=1$ we have an evidence that Problem \ref{nicetostudy2} is relevant; in fact Section \ref{sec2} shows that there are important conjectures on the deformation theory of finite dimensional nilpotent Lie algebras. Of course,  the general perspective which we suggest in Problem \ref{nicetostudy2} is more subtle to investigate and we leave it for successive contributions.

 \newpage

\section{Appendix I -  Schur multipliers of finite dimensional complex Lie algebras}\label{appendix1}

We report some material which is presented in \cite{bagrus2018, bagrus2019, bagrus2020, snobwin, weibel}, since it is useful for the nature of the considerations which have been done in the present paper. In fact we use several notions which belong to different areas of pure and applied mathematics, so it is useful for the reader to have them in the present appendix. Furthermore, we  list a series of properties which are designed specifically for a finite dimensional complex Lie algebra $\mathfrak{g}$ 
possessing  center \begin{equation}Z(\mathfrak{g})=\{ x \in \mathfrak{g}  \ | \ [x,y]=0  \ \forall y \in \mathfrak{g}\} 
\end{equation} and derived Lie subalgebra \begin{equation}
[\mathfrak{g}, \mathfrak{g}]= \langle [x,y] \ | \ x,y \in \mathfrak{g}\rangle. \end{equation}
Of course, when we deal with Lie algebras we have that $[x,y]=xy-yx$
and  $Z(\mathfrak{g})$ is an abelian ideal of $\mathfrak{g}$, in particular $Z(\mathfrak{g})$  is a finite dimensional vector space with complex coefficients.

\begin{definition}[See \cite{weibel}, Definitions 7.1.6 and 7.1.7]\label{uppercentralseries}The \textit{upper central series} of a finite dimensional complex Lie algebra $\mathfrak{g}$ is the ascending series of ideals
\begin{equation}0=Z_0(\mathfrak{g}) \leq Z_1(\mathfrak{g})\subseteq Z_2(\mathfrak{g}) \subseteq \ldots \subseteq Z_i(\mathfrak{g}) \subseteq Z_{i+1}(\mathfrak{g}) \subseteq \ldots, \end{equation}
where each $Z_i(\mathfrak{g})$ is called  $i$th $center$ of $\mathfrak{g}$ and is defined by \begin{equation} \frac{Z_1(\mathfrak{g})}{Z_0(\mathfrak{g})}=Z(\mathfrak{g}),\quad \frac{Z_2(\mathfrak{g})}{Z_1(\mathfrak{g})}=Z\left(\frac{\mathfrak{g}}{Z_1(\mathfrak{g})}\right),   \ldots, \frac{Z_{i+1}(\mathfrak{g})}{Z_i(\mathfrak{g})}=Z\left(\frac{\mathfrak{g}}{Z_i(\mathfrak{g})}\right), \ldots \end{equation}
We say that $\mathfrak{g}$ is  \textit{nilpotent of class $c$} if  the upper central series of $\mathfrak{g}$ ends after   $c $  steps. By analogy,  the \textit{lower central series} of $\mathfrak{g}$ is  defined in terms of commutators by the descending series of ideals
\begin{equation}\label{gammasequences}
\gamma_1(\mathfrak{g})=\mathfrak{g} \supseteq \gamma_2(\mathfrak{g})=[\mathfrak{g},\mathfrak{g}]  \supseteq  \gamma_3(\mathfrak{g})=[[\mathfrak{g},\mathfrak{g}],\mathfrak{g}]  \supseteq \ldots \supseteq \gamma_i(\mathfrak{g}) \supseteq \gamma_{i+1}(\mathfrak{g}) \supseteq \ldots, \end{equation}
where  $\gamma_i(\mathfrak{g})$ is called $i$th $derived$ of $\mathfrak{g}$ 
 and $\mathfrak{g}$ is  \textit{nilpotent of class $c$} if $\gamma_{c+1}(\mathfrak{g})=0$. The \textit{derived series} of  $\mathfrak{g}$ is the descending series of ideals 
\begin{equation}\label{powersequences}
\mathfrak{g} \supseteq [\mathfrak{g},\mathfrak{g}]=\mathfrak{g}^{(1)} \supseteq  [\mathfrak{g}^{(1)},\mathfrak{g}^{(1)}]=\mathfrak{g}^{(2)} \supseteq \ldots \supseteq \mathfrak{g}^{(i+1)}=[\mathfrak{g}^{(i)},\mathfrak{g}^{(i)}] \supseteq \ldots \end{equation}
where each  $\mathfrak{g}^{(i+1)}/\mathfrak{g}^{(i)}$ is an abelian Lie algebra and $\mathfrak{g}$ is  \textit{solvable of  length $m$} if    $\mathfrak{g}^{(m)}=0 $. 
 \end{definition}

If \eqref{gammasequences} stops after finitely many steps, then also \eqref{powersequences} makes the same, see \cite[Lemma 7.1.8]{weibel}, but the converse is false, since finite dimensional complex Lie algebras of upper triangular matrices give examples of solvable Lie algebras which are not nilpotent.

The  classification of a nonnilpotent finite dimensional complex Lie algebra $\mathfrak{a}$ for high dimensions may be done involving the classification of the nilpotent case when these are available. In fact several authors (see \cite{gokh, hofmor, snobwin}) use  the notion of \textit{nilradical (or nilpotent radical)} $\mathrm{Rad}_{\mathcal{N}}(\mathfrak{a})$ of $\mathfrak{a}$, i.e.: the largest nilpotent ideal in $\mathfrak{a}$, and then classify nonnilpotent finite dimensional Lie algebras looking at  the dimension of $\mathrm{Rad}_{\mathcal{N}}(\mathfrak{a})$  and at restrictions on the quotient Lie algebra  $\mathfrak{a}/\mathrm{Rad}_{\mathcal{N}}(\mathfrak{a})$.   In the same vein, one can classify nonsolvable finite dimensional complex Lie algebra $\mathfrak{b}$ involving the largest solvable ideal in  $\mathfrak{b}$,  called \textit{solvable radical} and  denoted by $\mathrm{Rad}_{\mathcal{S}}(\mathfrak{b})$. Of course,  $\mathrm{Rad}_{\mathcal{N}}(\mathfrak{b}) \subseteq \mathrm{Rad}_{\mathcal{S}}(\mathfrak{b})$, and in general
\begin{equation} \mathfrak{b} \  \mbox{ is nilpotent}  \ \ \Longleftrightarrow \ \ \mathrm{dim} \ \ \frac{\mathfrak{b}}{\mathrm{Rad}_{\mathcal{N}}(\mathfrak{b})}  =0 \ \ \Longleftrightarrow \ \ \mathfrak{b}=\mathrm{Rad}_{\mathcal{N}}(\mathfrak{b})
\end{equation}
\[\mathfrak{b} \  \mbox{ is solvable}  \ \Longleftrightarrow \ \ \mathrm{dim} \  \ \frac{\mathfrak{b}}{\mathrm{Rad}_{\mathcal{S}}(\mathfrak{b})}  =0  \ \ \Longleftrightarrow \ \ \mathfrak{b}=\mathrm{Rad}_{\mathcal{S}}(\mathfrak{b})
\]

\begin{definition}[See \cite{weibel}, Definition 7.6.1]\label{extensions}
We say that a finite dimensional complex Lie algebra $\mathfrak{g}$ is an \textit{extension} of an abelian Lie algebra $\mathfrak{a}$ (called \textit{abelian kernel}) by another  Lie algebra $\mathfrak{b}$, if there is a short exact sequence
$0 \ \longrightarrow \  \mathfrak{a} {\overset{\alpha}{\longrightarrow}} \ \mathfrak{g} \ {\overset{\beta}{\longrightarrow}} \mathfrak{b} \ \longrightarrow \ 0$
such that  $\mathfrak{a}$ is an abelian ideal of $\mathfrak{g}$ and $\mathfrak{g}/\mathfrak{a} \simeq \mathfrak{b}$ is the Lie algebra quotient; i.e.:  $\alpha$ is  monomorphism of Lie algebras,  $\beta$ is  epimorphism of Lie algebras and $\ker \beta = \mathrm{im} \ \alpha $. If in addition $\mathfrak{a} \subseteq Z(\mathfrak{g})$,  then $\mathfrak{g}$ is called \textit{central extension} of $\mathfrak{a}$ by $\mathfrak{b}$.
\end{definition}

Among extensions with abelian kernel, we find the so-called  \textit{split extensions} (see \cite[Exercise 7.6.1, Part 2]{weibel}), that is,  Lie algebras $\mathfrak{g}$ possessing two Lie subalgebras $\mathfrak{a}$ and $\mathfrak{b}$  such that $\mathfrak{a}$  is an ideal of $\mathfrak{g}$, $\mathfrak{g}= \mathfrak{a} + \mathfrak{b}$ and  $\mathfrak{a} \cap \mathfrak{b} = 0$. The relevance of finite dimensional  complex  Lie algebras, which are nilpotent or solvable,   has been investigated for a long time in mathematical physics, see \cite{sergey2, snobwin}.  Definitively, the first examples, which have been studied in the literature,  are the following:

\begin{definition}[See \cite{snobwin}, Example 5.12]\label{Heisenberg} We say that a finite dimensional complex Lie algebra $\mathfrak{g}$ is a $Heisenberg$ $Lie$ $algebra$ (or briefly a $Heisenberg$ $algebra$)   if \begin{equation}[\mathfrak{g},\mathfrak{g}]=Z(\mathfrak{g}) \  \ \mbox{and} \ \ \mathrm{dim} \
 [\mathfrak{g},\mathfrak{g}] = 1.
\end{equation} Such algebras are odd dimensional with
basis $v_1, \ldots , v_{2m}, v$ (with $m \ge1$ integer) and the only nonzero commutator between basis elements is $[v_{2i-1}, v_{2i}] = -
[v_{2i}, v_{2i-1}]= v$ for $i = 1,2, \ldots ,m$. They are  denoted by $\mathfrak{h}(m)$
and $\mathrm{dim} \ \mathfrak{h}(m)=2m + 1$.
\end{definition}

Comparing Definition \ref{Heisenberg} with \cite[Example 5.67, Exercise E6.9]{hofmor}, we note a large presence of abelian ideals  and abelian quotients in $\mathfrak{h}(m)$, but also a rich structure at the level of those Lie groups which can be naturally associated to  $\mathfrak{h}(m)$. The Heisenberg algebras play a fundamental role in several aspects of Lie theory, in fact they can be used to classify all low dimensional Lie algebras, see \cite[Chapters 15, 16, 17, 18, 19]{snobwin}. For this scope, we  mention  an approach of Morozov \cite{morozov}, as formulated in \cite[\S 8.1]{snobwin}, in order to classify finite dimensional nilpotent complex Lie algebras, but also another one of Skjelbred and Sund \cite{sund} as presented in \cite{bagrus2020}, since it is more appropriate to our context. To this scopes, we must introduce some concepts of homological algebra, namely  \cite[Definition 7.7.1, Exercises 7.7.4, 7.7.5]{weibel} and the notion of \textit{Chevalley-Eilenberg complex}. 
 
\begin{definition}[See \cite{snobwin}, Chapter 2]\label{ce1}Let $V$ be a complex vector space of finite dimension $n$ (with $n \ge 1$ integer) and $\mathfrak{gl}_n(\mathbb{C})$   the Lie algebra of the linear operators acting on $V$ with Lie bracket $[f,g]=f  g -g f$ for all $f,g \in \mathfrak{gl}_n(\mathbb{C})$. A  \textit{representation} $\rho$  \textit{of degree} $n$  of a   complex Lie algebra $\mathfrak{g}$  of finite dimension $\mathrm{dim} \ \mathfrak{g}$ is a linear map $\rho : x\in \mathfrak{g} \mapsto \rho(x) \in \mathfrak{gl}_n(\mathbb{C})$  preserving the Lie bracket $[ \ , \ ]$ of $\mathfrak{g}$, i.e.: $\rho([y,z])=\rho(y)  \rho(z) - \rho(z) \rho(y)$ for all $y,z \in \mathfrak{g}$, and  $\rho$ is called \textit{faithful} if it is injective.  If $\mathrm{dim} \ \mathfrak{g}=n$,  a particular representation is given by \textit{the adjoint}, which is the map $\mathrm{ad} : x \in \mathfrak{g} \longmapsto \mathrm{ad} (x)  \in \mathrm{Aut}(\mathfrak{g})$ defined by $\mathrm{ad} (x) : y  \in \mathfrak{g} \mapsto [x,y] \in \mathfrak{g}$, where $\mathrm{Aut}(\mathfrak{g})$ denotes the Lie algebra of the automorphisms of $\mathfrak{g}$.

The $k$-\textit{cochain} of $\mathfrak{g}$ with complex coefficients is an alternating multilinear map  \begin{equation}c : (x_1, x_2, \ldots, x_k) \in \underbrace{\mathfrak{g} \times  \mathfrak{g} \times \ldots \times \mathfrak{g}}_{k-\mbox{times}} \mapsto c(x_1, x_2, \ldots, x_k) \in V, 
\end{equation} where $k\in \{0,1,\ldots,\mathrm{dim} \ \mathfrak{g}\}$. The set of all  $k$-cochains $C^k (\mathfrak{g}, V;\rho)$ is a  vector space, so we may define the \textit{cochain complex of} $\mathfrak{g}$   \begin{equation}\label{chevcomplex} C^\bullet(\mathfrak{g},V;\rho)= \bigoplus^{\mathrm{dim} \ \mathfrak{g}}_{k=1}C^k(\mathfrak{g},V;\rho) 
\end{equation}  and this is another vector space.  The \textit{cohomology operator}  of $\mathfrak{g}$ with complex coefficients (or \textit{differential of k-cochain}) is defined by
\begin{equation}\label{cohomologyoperator}
\mathrm{d} \ : \ c(x_1, x_2, \ldots, x_{k+1}) \in  C^{k+1} (\mathfrak{g}, V;\rho) \longmapsto \mathrm{d} (c (x_1, x_2, \ldots, x_{k+1}))
\end{equation}
\[= \sum^{k+1}_{i=1}{(-1)}^{i+1}  \rho(x_i) c (x_1, x_2, \ldots, \widehat{x}_i ,\ldots, x_{k+1})\]
\[ + \sum_{1\le i<j \le k+1}{(-1)}^{i+j}  c ([x_i,x_j], x_1, \ldots, \widehat{x}_i ,\ldots, \widehat{x}_j, \ldots, x_{k+1}) \in C^k (\mathfrak{g}, V;\rho)\]
(with the agreement that $\widehat{x}_i $ indicates  the term which is
omitted at the $i$th coordinate) and extended  linearly  to  $C^\bullet(\mathfrak{g},V;\rho)$. A $k$-cochain $c$ is a $k$-\textit{cocycle} if $\mathrm{d}  \ c=0$, i.e.:  $\mathrm{d}$ is constant to $0$ on $c$.
\end{definition}

The cohomology operator $\mathrm{d}$ is an involution, that is, $\mathrm{d}^2 =0$ (see \cite{snobwin, weibel}) and this allows us to introduce further important concepts. Note  that \eqref{chevcomplex} of Definition \ref{ce1} is also called \textit{Chevalley–Eilenberg complex} and we are in fact describe the so called \textit{Chevalley–Eilenberg cohomology} of  finite dimensional complex Lie algebras.

\begin{definition}[See  \cite{snobwin}, Chapter 2]\label{ce2} In line with Definition \ref{ce1}, the set  $Z^k(\mathfrak{g},V;\rho)$ of all $k$-cocycles on $\mathfrak{g}$ forms an abelian Lie algebra of finite dimension and  
\begin{equation}
Z^k(\mathfrak{g},V;\rho) = \ker \ \mathrm{d}  \ \cap \  C^k(\mathfrak{g},V;\rho).
\end{equation}
A $k$-cochain $c$ is called $k$-\textit{coboundary} if $c = \mathrm{d} \ c'$ for some other ($k-1$)-cochain $c'$.  The set  $B^k(\mathfrak{g},V;\rho)$ of all $k$-coboundaries on $\mathfrak{g}$ forms a Lie subalgebra of $Z^k(\mathfrak{g},V;\rho)$ and 
\begin{equation}
B^k(\mathfrak{g},V;\rho) = \mathrm{im} \ \mathrm{d}  \ \cap \  C^k(\mathfrak{g},V;\rho).
\end{equation}
The $k$-\textit{cohomology} with complex coefficients of $\mathfrak{g}$ is defined by 
\begin{equation}\label{ce3}
H^k(\mathfrak{g},V;\rho) = \frac{Z^k(\mathfrak{g},V;\rho)}{B^k(\mathfrak{g},V;\rho)}.
\end{equation}
The \textit{Chevalley-Eilenberg complex} of $\mathfrak{g}$ is the  finite dimensional  vector space
\begin{equation}H^\bullet(\mathfrak{g},V;\rho)= \bigoplus^{\mathrm{dim} \ \mathfrak{g}}_{k=1}H^k(\mathfrak{g},V;\rho).
\end{equation} 
\end{definition}

There are important interpretations of \eqref{ce3} when $k=1$ or $k=2$.

\begin{remark}\label{derivations} Recall from \cite{hofmor, snobwin, weibel} that a \textit{derivation} $D$ of  $\mathfrak{g}$ is a linear map \begin{equation}D : x \in \mathfrak{g} \mapsto D(x) \in \mathfrak{g}
\end{equation} such that for any $x, y \in \mathfrak{g}$ we have \begin{equation}D([x, y]) = [D(x), y] + [x, D(y)].
\end{equation} If for all $x \in \mathfrak{g}$ there is an element $z \in \mathfrak{g}$  such that $D(x) = [z, x]$, we say that  $D$ is an \textit{inner derivation}. If this doesn't happen, $D$ is called \textit{outer derivation}. 
It can be seen that the set of all derivations $\mathrm{Der}(\mathfrak{g})$ of  $\mathfrak{g}$ forms a Lie algebra (in fact derivations are automorphisms of Lie algebras) and those which are inner form an ideal $\mathrm{Inn}(\mathfrak{g})$ of $\mathrm{Der}(\mathfrak{g})$. In particular, if \begin{equation}\mathrm{Rad}_{\mathcal{S}}(\mathrm{Der}(\mathfrak{g})) \neq 0, 
\end{equation}
we say that $\mathfrak{g}$ has \textit{nontrivial semisimple derivations}.
\end{remark}

\begin{remark}\label{derivations1}
Fix now $k=1$ in Definition \ref{ce2} and consider $H^1(\mathfrak{g},\mathfrak{g};\mathrm{ad})$, specifying the representation $\rho$ as the adjoint $\mathrm{ad}$. Now 1-cocycles are exactly the derivations of $\mathfrak{g}$ by definition, because $c([x,y])=[x,c(y)] +[c(x),y]$, and 1-coboundaries are inner derivations, because  for $x \in \mathfrak{g}$ and $\mathrm{d} x \in B^1(\mathfrak{g}, \mathfrak{g};\mathrm{ad})$ we get $(\mathrm{d} x)y= \mathrm{ad}(y) x= -\mathrm{ad}(x) y$, that is, $\mathrm{d} x=- \mathrm{ad}(x)$. Therefore
\begin{equation}H^1(\mathfrak{g},\mathfrak{g};\mathrm{ad}) = \frac{Z^1(\mathfrak{g},\mathfrak{g};\mathrm{ad})}{B^1(\mathfrak{g},\mathfrak{g};\mathrm{ad}) } = \frac{\mathrm{Der}(\mathfrak{g})}{\mathrm{Inn}(\mathfrak{g})}.
\end{equation}
In other words, $H^1(\mathfrak{g},\mathfrak{g};\mathrm{ad})$ gives information on the derivations of $\mathfrak{g}$ modulo those which are inner.  Note that the set of all automorphisms of $\mathfrak{g}$ with composition  forms what is called \textit{a Lie group}, namely $\mathrm{Aut}(\mathfrak{g})$, and its Lie algebra is exactly $\mathrm{Der}(\mathfrak{g})$, so $H^1(\mathfrak{g},\mathfrak{g};\mathrm{ad})=\mathrm{Out}(\mathfrak{g})$ is the abelian Lie algebra of the outer derivations. See also \cite{gokh, hofmor, weibel}.
\end{remark} 

\begin{remark}\label{derivation2}Fix $k=2$ in Definition \ref{ce2}.  The \textit{Schur multiplier}  of $\mathfrak{g}$ is defined by
\begin{equation} H^2(\mathfrak{g},\mathfrak{g};\mathrm{ad}) = \frac{Z^2(\mathfrak{g},\mathfrak{g};\mathrm{ad})}{B^2(\mathfrak{g},\mathfrak{g};\mathrm{ad}) }
\end{equation} 
and in this situation, that is, when we specify $\rho=\mathrm{ad}$, we can write briefly $M(\mathfrak{g})= H^2(\mathfrak{g})=Z^2(\mathfrak{g})/B^2(\mathfrak{g})$. This  is an abelian Lie algebra and  is strongly related to the extension theory of Lie algebras. In fact \cite[Theorem 7.6.3]{weibel} shows that there exists a bijective correspondence between an element of $M(\mathfrak{g}) $ and an appropriate class of equivalence of extensions of $\mathfrak{g}$, that is, computing 
 $M(\mathfrak{g}) $ we may control all possible extensions of $\mathfrak{g}$.
\end{remark}

\begin{remark}\label{derivation3} Fix again $k=2$ in Definition \ref{ce2}  and consider $\theta : (x , y) \in  \mathfrak{g} \times \mathfrak{g} \mapsto \theta(x , y)  \in V$  alternating bilinear map  satisfying for all $x,y,z \in \mathfrak{g}$ the  condition 
\begin{equation}\label{cycle}\theta([x,y],z) + \theta([z,x],y) + \theta([y,z],x)= 0.
\end{equation}
This means that $\theta \in Z^2(\mathfrak{g})$, that is, preserving the Jacobi identity as above is necessary and sufficient for $\theta$ to be a 2-cocycle.\end{remark}

According to \cite[Definition 7.8.1]{weibel},  a finite dimensional complex Lie algebra $\mathfrak{g}$
 is semisimple if and only if   $\mathfrak{g}$ has no nontrivial abelian ideals. This is equivalent to require that   $\mathrm{Rad}_{\mathcal{S}}(\mathfrak {g})=0$. Note that
the structure of semisimple Lie algebras  \cite[Theorem 7.8.5]{weibel} shows that a finite dimensional complex Lie algebra is semisimple if and only if  it can be decomposed in direct sum of finitely many  (complex) simple Lie algebras. Therefore it turns out to be relevant a classification for those finite dimensional complex Lie algebras which are nonsimple. An idea is illustrated below and has computational nature.

\begin{remark}[Method  of Skjelbred and Sund, see \cite{sund}]\label{crucialobservation}
Fix  $k=2$ in Definition \ref{ce2} and consider the Schur multiplier $M(\mathfrak{g})$ as in Remarks \ref{derivation2} and \ref{derivation3}. Given $\theta \in Z^2(\mathfrak{g}) $, the kernel of $\theta$ and the annihilator of the 2nd component of $\theta$ are respectively  
\begin{equation}\label{radical}
 \ker \theta =\{(x,y)  \ |  \ \theta (x,y)=0 \}, \ \ \   \theta^+ =\{x\in \mathfrak{g}  \ |  \ \theta (x,y)=0,  \ \ \forall y \in \mathfrak{g} \}.
\end{equation} 
 We may introduce a new set \begin{equation}\label{delicate}\mathfrak{g}_\theta = \mathfrak{g} \oplus V
\end{equation}
and endow it of the Lie algebra structure, defining for all $x, y  \in \mathfrak{g}$ and $u,v \in V$, that is, for all $(x+u, y+v) \in \mathfrak{g}_\theta \times \mathfrak{g}_\theta$ the Lie bracket 
\begin{equation}\label{liebracket}
[x+  u, \ y + v] = {[x,y]}_\mathfrak{g} + \theta(x,y)
\end{equation}   on $\mathfrak{g}_\theta,$ where ${[ \ , \ ]}_\mathfrak{g}$ denotes the Lie bracket operation in $\mathfrak{g}$.

We are basically extending the Lie bracket in $\mathfrak{g}$ with an additive term in \eqref{liebracket}, depending only on $\theta$. Compare with Definition \ref{lindef}. One can easily check that \eqref{liebracket} is a new Lie bracket, but this time in  $\mathfrak{g}_\theta.$ Of course, Definition \ref{extensions} applies to $\mathfrak{g}_\theta$ and in fact $\mathfrak{g}_\theta$ is a  new complex Lie algebra of finite dimension which appears as central extension of $\mathfrak{g}$.  It is also elementary to check that $\mathfrak{g}_\theta \simeq \mathfrak{g}_{\theta + \mu}$ for all $\mu \in B^2(\mathfrak{g})$, that is, the construction of the central extension is unique up to coboundaries, see details in \cite{bagrus2020, sund, snobwin}. Therefore, one can assume $\theta \in M(\mathfrak{g})$ without loss of generality, in order to form central extensions of $\mathfrak{g}$. This  process holds in particular when we replace the role of $\mathfrak{g}$ with $\mathfrak{g}/Z(\mathfrak{g})$ and that of $V$ with $Z(\mathfrak{g}) \neq 0$ beginning from \eqref{delicate}. In fact any finite dimensional complex Lie algebra with a nontrivial centre can be obtained as a central extension of a Lie algebra of smaller dimension in this way. \end{remark}

There are alternative methods to that of Skjelbred and Sund; long time ago  Dixmier \cite{dixmier} and Burde \cite{burde1, burde2} used the main concepts in a noncomputational approach. In particular, all finite dimensional complex nilpotent Lie algebras can be constructed  with the methods of Skjelbred, Sund, Dixmier, Burde and other authors: this justifies from another point of view  possible positive answers to Conjectures \ref{goh} and  \ref{vc}.

It can be useful to recall, as matter of evidence,  that the classification of finite dimensional nilpotent  complex Lie algebras is  known for lower dimensional cases and it is reported below up to dimension 5, see details in \cite{seeley1, seeley2, snobwin}. Unfortunately, the classification becomes more difficult when the dimension is higher, so one has to introduce different techniques.

\begin{proposition}[Classification of Finite Dimensional Nilpotent Lie Algebras of Dimension 3, 4 and  5, see \cite{snobwin}]\label{classification}
Let $\mathfrak{n}$ be a finite dimensional complex nilpotent Lie algebra  and $\mathfrak{i}$ an abelian Lie algebra of dimension 1. Then
\begin{itemize}
\item[(1).] $\mathrm{dim} \ \mathfrak{n} =3$ if and only if $\mathfrak{n}$ is isomorphic to one of the following Lie algebras:
 \begin{itemize}
\item[-] $\mathfrak{n}_{3,1}$ abelian Lie algebra of dimension 3,
\item[-] $\mathfrak{n}_{3,2}  \simeq \mathfrak{h}(1)$.
\end{itemize}
\item[(2).]$\mathrm{dim} \ \mathfrak{n} =4$ if and only if    $\mathfrak{n}$ is isomorphic to one of the following Lie algebras:
 \begin{itemize}
\item[-] $\mathfrak{n}_{4,1} =  \mathfrak{n}_{3,1} \oplus \mathfrak{i}$,
\item[-] $\mathfrak{n}_{4,2} =  \mathfrak{n}_{3,2} \oplus \mathfrak{i}$
\item[-] $\mathfrak{n}_{4,3} = \langle v_1, v_2, v_3, v_4 \ | \ [v_1,v_2]=v_3, [v_1,v_3]=v_4 \rangle$
\end{itemize}
\item[(3).]$\mathrm{dim} \ \mathfrak{n} =5$ if and only if $\mathfrak{n}$ is isomorphic to one of the following Lie algebras:
\begin{itemize}
\item[-] $\mathfrak{n}_{5,1} = \mathfrak{n}_{4,1} \oplus \mathfrak{i} $,
\item[-] $\mathfrak{n}_{5,2} = \mathfrak{n}_{4,2} \oplus \mathfrak{i} \simeq \mathfrak{h}(1) \oplus \mathfrak{i} \oplus \mathfrak{i}$,
\item[-] $\mathfrak{n}_{5,3} = \mathfrak{n}_{4,3} \oplus \mathfrak{i} $,
\item[-] $\mathfrak{n}_{5,4} =  \langle v_1, v_2, v_3, v_4, v_5 \ | \ [v_1,v_2]=[v_3,v_4]=v_5 \rangle \simeq \mathfrak{h}(2) $,
\item[-] $\mathfrak{n}_{5,5} =  \langle v_1, v_2, v_3, v_4, v_5 \ | \ [v_1,v_2]=v_3, [v_1,v_3]=[v_2,v_4]=v_5 \rangle$,
\item[-] $\mathfrak{n}_{5,6} =  \langle v_1, v_2, v_3, v_4, v_5 \ | \ [v_1,v_2]=v_3, [v_1,v_3]=v_4,$ $$[v_1,v_4]=[v_2,,v_3]=v_5  \rangle,$$
\item[-] $\mathfrak{n}_{5,7} =  \langle v_1, v_2, v_3, v_4, v_5 \ | \ [v_1,v_2]=v_3, [v_1,v_3]=v_4, [v_1,v_4]=v_5 \rangle$,
\item[-] $\mathfrak{n}_{5,8} =  \langle v_1, v_2, v_3, v_4, v_5 \ | \ [v_1,v_2]=v_4, [v_1,v_3]=v_5 \rangle$,
\item[-] $\mathfrak{n}_{5,9} =  \langle v_1, v_2, v_3, v_4, v_5 \ | \ [v_1,v_2]=v_3, [v_1,v_3]=v_4, [v_2,v_3]=v_5 \rangle$.
\end{itemize}
\end{itemize}
\end{proposition}

Note that for nonsolvable (and also for nonnilpotent) Lie algebras we have:

\begin{proposition}[See Theorem 7.8.13, Levi Decomposition, in \cite{weibel}]
If $\mathfrak{g}$ is a finite dimensional complex Lie algebra, then $\mathfrak{g}$ is the split extension of $\mathfrak{a}=\mathrm{Rad}_{\mathcal{S}}(\mathfrak{g})$ which is solvable by $\mathfrak{b} = \mathfrak{g}/\mathrm{Rad}_{\mathcal{S}}(\mathfrak{g})$ which is semisimple.
\end{proposition}

Therefore Remark \ref{crucialobservation} in addition to the factorization of semisimple Lie algebras in direct sum of simple Lie algebras provide a complete structural description for arbitrary finite dimensional complex Lie algebras. We end with a few information of differential geometry, in connection with the affine varieties of complex finite dimensional Lie algebras (in particular this is true in the nilpotent case).

\begin{remark} We have that  $\mathcal{O}(\mu) = 
\mathrm{GL}_n(\mathbb{C})/{\mathrm{GL}_n(\mathbb{C})}_\mu$   in \eqref{orbitdef} possesses the structure of  differentiable manifold (see details in \cite{grunewald1, grunewald2, rem1, seeley1, seeley2}). In particular, it turns out that its tangent space $T_\mu  \mathcal{O}(\mu)$ to $\mu$ is well defined and  isomorphic to the quotient space $\mathfrak{gl}_n(\mathbb{C})/\mathrm{Der}(\mu)$ is the  Lie algebra of the derivations of  ${\mathrm{GL}_n(\mathbb{C})}_\mu$.  We deduce that
 \begin{equation}\dim  \ T_\mu \mathcal{O}(\mu)=n^2-\dim  \ \mathrm{Der} (\mu)\end{equation}
 and so $T_\mu \mathcal{O}(\mu)$ is isomorphic to the subspace of bilinear maps whose elements are coboundaries. In particular, if we have a complex finite dimensional Lie algebra $\mathfrak{g}$ with Lie bracket $\mu$,  then  we may detect that $\mu$ is rigid  whenever 
 \begin{equation}\dim  \ T_\mu \mathcal{O}(\mu)=\dim  \ B^2(\mathfrak{g}),\end{equation}
as illustrated by \cite[Proposition 5]{rem1}.
\end{remark}

We end, noting that Carles \cite{carles} investigated the structure of rigid Lie algebras over algebraically closed fields of characteristic zero. In particular he proved that nilpotent Lie algebras of rank $\ge 1$ are never rigid and moreover
nilpotent Lie algebras with a codimension 1 ideal of rank $\ge 1$ are also never rigid.
In other words, Conjecture \ref{vc} holds for this class, remaining open for characteristically nilpotent Lie algebras for which all its ideals of codimension 1 are also characteristically nilpotent. This is an open problem at the moment.

\newpage

\section{Appendix II - Low dimensional cohomology of $q$-deformed Heisenberg algebras}\label{qdeformedappendix}

We  formalized $\mathcal{H}(q)$ in Definition \ref{qh}, so we have the notion of $q$-deformed Heisenberg algebra, where we assume $q \in [-1,1]$. 
Note that  $\mathcal{H}(q)$  is a particular type of associative algebra with unit. By a simple routine calculation based on a well-known result in the literature about free Lie algebras, the relation $AB-qBA=\mathbb{I}$, which is typical of $q$-deformed Heisenberg algebras,  does not satisfy the Jacobi identity when $q\neq 1$,  so it is necessary to generalize the usual rules of the commutator calculus on Lie algebras. In fact the  $q$-mutator $[ \ , \ ]_q$ of \eqref{qmutator} does not give  a Lie bracket when $q\neq 1$, but we can still say something when we consider associative (complex) algebras such as $\mathcal{H}(q)$ satisfying the rule $[A,B]_q=\mathbb{I}$. 
Of course, ${[ \ , \ ]}_1=[ \ , \ ]$ so we omit $q=1$ when we refer to usual Lie brackets.

\begin{remark}If  $q=1$, then $\mathcal{H}(1)$ is the three dimensional complex nilpotent Lie algebra  with basis $A,B,\lbrack B,A\rbrack$ that satisfy the rules $[A,B]=\mathbb{I}$,  $\lbrack \lbrack B,A\rbrack,A\rbrack=0$ and $\lbrack B, \lbrack B,A\rbrack\rbrack =0$. In fact we have $\mathcal{H}(1) =\mathfrak{h}(1)$,  see \cite[Section 1.1]{sergey2} and Proposition \ref{classification} (1). \end{remark}

One can also show that for $q \neq 1$ the $q$-deformed Heisenberg algebra $\mathcal{H}(q)$ possesses a large  Lie subalgebra $\mathcal{L}(q)$ and most of the times $\mathcal{L}(q)$ turns out to be an ideal, so that $\mathcal{H}(q)/\mathcal{L}(q)$ becomes easier to manage and $\mathcal{H}(q)$ might be an extension of $\mathcal{L}(q)$ by $\mathcal{H}(q)/\mathcal{L}(q)$, see details in \cite{cantuba, sergey1, sergey2, sergey3}. We just give an idea when $q=0$.

\begin{proposition}[The Case of $\mathcal{H}(0)$, see \cite{cantuba}, Theorem 4.10] There exists a Lie subalgebra $\mathcal{L}(0)$ in $\mathcal{H}(0)$ which is an ideal. Moreover  the quotient algebra $\mathcal{H}(0)/\mathcal{L}(0)$ is an   abelian Lie algebra.
\end{proposition}

Let's go ahead and recall some useful rules of the  \emph{$q$-special combinatorics} \cite[Appendix C]{sergey2}. This is suitable for the context of $\mathcal{H}(q)$, so for the main topics which are discussed in the present paper.

\begin{proposition}[$q$-Combinatorial Identities, see \cite{sergey2}, Theorems C.6, C.7, C.9, C.10]\label{generealizedjacobi1}For  integers $n,k \ge 1$   define for  $ q \in [-1,1]$ 
\begin{equation}
\{n\}_q  =  \sum_{l=0}^{n-1} q^l, \ \  \ \ \
\{n\}_q!  =  \prod_{l=1}^n\{ l\}_q \ \ \ \  \mbox{and} \ \ \ \
{{n}\choose{k}}_q =  {{n-1}\choose{k-1}}_q + q^k {{n-1}\choose{k}}_q \ \ \mbox{when} \ k < n,
\end{equation}
where 
\begin{equation}    {{n}\choose{k}}_q =0  \ \ \mbox{if} \ \ k>n  \ \ \ \ \ \mbox{and}  \ \  \ \ \ {{n}\choose{k}}_q =1 \ \   \ \ \mbox{if}   \ \ k=n.
\end{equation} Then for all $k \le n$ we have that
\begin{equation}
{{n}\choose{k}}_q=\frac{\{n\}_q!}{\{k\}_q!\{n-k\}_q!},  \ \ \   {{n}\choose{k}}_q = {{n}\choose{n-k}}_q  \ \ \mbox{and} \ \ {{n}\choose{k}}_q={\underset{S \subseteq \{1,2, \ldots, n\}}{\underset{|S|=k}\sum }}  \frac{2q^{\underset{s \in S} \sum s} }{ k (k+1)}.
\end{equation}
Moreover for $q \neq 0$ we also have
\begin{equation}
\{n\}_\frac{1}{q}  =  \frac{q}{q^{n}} \{n\}_q, \ \  
\{n\}_\frac{1}{q} !  =  \frac{1}{q^{{n}\choose{2}}} \{n\}_q \ \   \mbox{and} \ \ 
{{n}\choose{k}}_\frac{1}{q} =  \frac{1}{q^{k (n-k)}} {{n}\choose{k}}_q.
\end{equation}
\end{proposition}
One of the most important properties of $\mathcal{H}(q)$ is that elements of the form $B^mA^n$ produce a basis for $\mathcal{H}(q)$, according with \cite[Theorem 3.1]{sergey2}. We may express them via  rules of  $q$-calculus.

\begin{proposition}[Powers and Products of Generators, see \cite{sergey2}, Theorem 2.2]\label{powandprod} If $n\ge 1$ and $A,B$ are two elements of an associative algebra on $\mathbb{C}$ with unit $\mathbb{I}$, then 
\begin{equation} 
AB^n  =  q^nB^nA + \{n\}_q B^{n-1} \ \ \mbox{and} \ \ 
A^nB  =  q^nBA^n + \{n\}_q A^{n-1}, \ \forall q \in [-1,1] \end{equation}
\begin{equation}BA^n  =  \frac{1}{q^n} A^nB - \frac{1}{q}\{ n\}_{\frac{1}{q}}A^{n-1} \ \ \mbox{and} \ \ 
B^nA  =  \frac{1}{q^n} AB^n - \frac{1}{q}\{ n\}_{\frac{1}{q}} B^{n-1}, \ \forall q \in [-1,1] \setminus \{0\}.\end{equation}
\end{proposition}
%It can be useful to think at the product $HK$ of two elements $H,K$ of $\mathcal{H}(q)$ in terms of  linear combinations, that is, it is useful to put \begin{equation}HK:=\mathrm{span} \{hk\  |\    h\in H,\  k\in K\}.\end{equation} 
%\begin{proposition}[Bases for $\mathcal{H}(q)$,  see \cite{sergey2}]\label{generalizedjacobi2} If $q\notin\{0,1\}$, then the vectors \begin{equation}{[A,B}]^k, \  B^l {[A,B]}^k,\  [A,B]^k A^l,\quad \mbox{where}  \ k,  l \ge 0, \end{equation}
%form a basis for $\mathcal{H}(q)$. 
%\end{proposition}
%Furthermore one can make computations and find the following simpler bases.
%\begin{remark} Looking at Propositions \ref{powandprod} and \ref{generalizedjacobi2}, we have
%\begin{itemize}
%\item[\rm{(i).}] ${[A,B]}^k$ with $k \in \mathbb{N}$ form a basis for  $\mathrm{span}\{B^kA^k\ %|\  k \ge 0\} $ of $\mathcal{H}(q)$;
%\item[\rm{(ii).}]  $B^l{[A,B]}^k$  with $k, l \ge 0$ form a basis for $\mathrm{span} \{B^kA^l\ |\  l\in \mathrm{N},  k=2l\}$ of $\mathcal{H}(q)$;
%\item[\rm{(iii).}]
% ${[A,B]}^kA^l$ with $k\ge 0$ form a basis for  $\mathrm{span} \{A^l\ |\  l\ge 0 \}$.
%\end{itemize}
%\end{remark}

Additional rules can be obtained,  but the computational difficulties become significant, so we refer to \cite[Chapters 3 and 7]{sergey2} for more details. Continuing to present some rules for $q$-calculus which are useful to our scope, we mention the following results.

\begin{proposition}[Generalized Jacobi Identities, see \cite{cantuba}, Propositions 3.2, 3.3, Lemma 3.4] \label{generalizedjacobi3} Given an integer $n \ge 1$  and $q\in [-1,1]$,   two elements  $A,B$ of an associative algebra on $\mathbb{C}$ with unit $\mathbb{I}$ satisfy the following equations:
\begin{equation}
B^nA^n = \frac{1}{q^{{n}\choose{2}} \ {(q-1)}^n}\sum_{k=0}^n(-1)^{n-k}q^{{n-k}\choose{2}}{{n}\choose{k}}_q  \lbrack A,B\rbrack^k \ \ \ \forall q \not\in   \{0,1\}, 
\end{equation}
\begin{equation}
A^nB^n = \frac{1}{(q-1)^n}\sum_{i=0}^n(-1)^{n-k}q^{{k+1}\choose{2}}{{n}\choose{k}}_q\lbrack A,B\rbrack^k \ \ \ \forall q \not\in   \{0,1\}, 
\end{equation}
\begin{equation}
\lbrack B,B^mA^n\rbrack  =   (1-q^n)B^{m+1}A^n-\{n\}_qB^mA^{n-1} \ \ \ \forall m \ge 1. 
\end{equation}
\end{proposition}

Note that  $q=0$ implies
$AB^n  =  B^{n-1}$ and $ A^nB  =  A^{n-1}.$ In particular, it is possible to extend Proposition \ref{generalizedjacobi3} to the case of $q=0$ for nonzero integers $m,n$, just writing that
\begin{eqnarray}
A^nB^m  = \left\{
\begin{array}{ll} A^{n-m},  n\geq m,\\
B^{m-n},  n<m.\end{array}\right.
\label{ferAnBm}
\end{eqnarray}
We omit details, since they can be found in \cite[Section 4]{cantuba}. Proposition \ref{generalizedjacobi3} may be applied also to the case of $q=1$, but we refer to \cite{cantuba} for details. Note that Proposition \ref{generalizedjacobi3} would give for $q=1$ the usual rules of commutators and Lie brackets. Incidentally, it is worth to mention that Proposition \ref{generalizedjacobi3} offers a first generalization of the usual Jacobi Identity in the $q$-calculus, but we will describe better this aspect in the following result.
%\begin{remark}
%A further simplification when $q=0$ is given by
%\begin{equation}[B,B^mA^n]  =  B^{m+1}A^n-B^mA^{n-1},\label{adB0}
%\end{equation} 
%\begin{equation}
%[B,A^n]  =  BA^n-A^{n-1}, 
%\end{equation}
%\begin{equation}
% [B^m,A]  =  B^mA-B^{m-1}. 
%\end{equation}
%In fact the following relations hold in $\mathcal{H}(0)$ for any $i,j,k,l\geq 4$.
%\begin{equation}\lbrack B^iA^2,B^lA^2\rbrack   = \lbrack B^2A^j, B^2A^k\rbrack =0 .\label{idLemeq5}
%\end{equation}
%\end{remark}
Some consequences of Propositions \ref{generealizedjacobi1}, \ref{powandprod} and \ref{generalizedjacobi3} are in fact summarized below.

\begin{corollary}\label{generalizedjacobi4}
If   $A,B, C$ are three elements of an associative algebra on $\mathbb{C}$ with unit $\mathbb{I}$,    $q \in [-1,1]$ and $\alpha, \beta \in \mathbb{C}$, then   
 \begin{itemize}
     \item[{\rm (i).}] the trace satisfies the invariant property $\mathrm{tr} ({[A,B]}_q) = \mathrm{tr} ({[B,A]}_q)$;
     \item[{\rm (ii).}] the map ${[ \ , \ ]}_q : \mathcal{H}(q) \times \mathcal{H}(q) \to \mathcal{H}(q)$ is bilinear, that is, ${[ A+C, B ]}_q = {[ A, B ]}_q + {[ C, B ]}_q$, ${[ A, B + C]}_q = {[ A, B ]}_q + {[ A, C ]}_q$ and ${[ \alpha A, \beta B ]}_q = (\alpha \beta) \  {[ A, B ]}_q $;
     \item[{\rm (iii).}] ${[ \ , \ ]}_q$  satisfies the  first $q$-antisymmetric property ${[ A , B ]}_q = - q  \ {[ B , A ]}_\frac{1}{q}$ when $q\neq 0$; 
     \item[{\rm (iv).}] ${[ \ , \ ]}_q$  satisfies the  second $q$-antisymmetric property ${[ B , A ]}_q =    {[ A , B ]}_q - (1 + q) \ [ A , B ]$;
     \item[{\rm (v).}] ${[ \ , \ ]}_q$  satisfies a $q$-Jacobi Identity of type one ${[AB,C]}_q = [A,BC] + {[B, CA]}_q$;
     \item[{\rm (vi).}] ${[ \ , \ ]}_q$  satisfies a $q$-Jacobi Identity of type two ${[A, BC]}_q = {[AB,C]}_q + q[CA,B]$;
     \item[{\rm (vii).}]  ${[ \ , \ ]}_q$  satisfies a $q$-Jacobi Identity of type three
     \[{[ A , {[ B , C ]}_q ]}_q + {[ B , {[ C , A ]}_q ]}_q + {[ C , {[ A , B ]}_q ]}_q = (1-q) \  ( (XYZ  + YZX + ZXY) - q (XZY + YXZ + ZYX)).\]
 \end{itemize}
\end{corollary}

Note that Corollary \ref{generalizedjacobi4} has been proved when $A$, $B$ and $C$ are square complex matrices in \cite[Theorem 1, Lemma 2 and Theorem 3]{kuw}. In particular, note that for $q=1$ in Corollary \ref{generalizedjacobi4} we get the well known properties of the Lie bracket in the case of finite dimensional complex Lie algebras. Bases, generators and relations for $\mathcal{H}(q)$ are available in \cite[Theorem 3.1]{sergey2}.

 Let $V$ be a complex vector space such that $q \in [-1,1]$  and   \begin{equation} \mathrm{dim} \ \mathcal{H} (q)= |\{B^kA^l \mid k, l \in \mathbb{N} \}|. \end{equation} 
The map \begin{equation}\label{qadjoint} \mathrm{ad}_q : x \in \mathcal{H} (q) \longmapsto \mathrm{ad}_q (x)={[x,y]}_q  \in \mathrm{Aut}(\mathcal{H} (q))\end{equation} is the \textit{q-adjoint} of $x \in \mathcal{H} (q)$, where $\mathrm{Aut}(\mathcal{H} (q))$ denotes the associative algebra (with unit) of all invertible linear maps from $\mathcal{H} (q)$ to $\mathcal{H} (q)$. Note that Corollary \ref{generalizedjacobi4} (2) ensures  $\mathrm{ad}_q (x)  \in \mathrm{Aut}(\mathcal{H} (q))$, so \eqref{qadjoint} is well defined. 
The $2$-\textit{cochain} of $\mathcal{H} (q)$ with complex coefficients is a multilinear map  \begin{equation}c : (x_1, x_2) \in \mathcal{H} (q) \times  \mathcal{H} (q) \mapsto c(x_1, x_2) \in V. 
\end{equation}  The set of all  $2$-cochains $C^2 (\mathcal{H} (q), V)$ is a  vector space, so we may define in analogy $1$-cochains $C^1 (\mathcal{H} (q), V)$ and $3$-cochains $C^3 (\mathcal{H} (q), V)$, in order to get  cohomology operators $\mathrm{d}_2 \ : C^2 (\mathcal{H} (q),V) \to C^1 (\mathcal{H} (q),V)$ for $\mathcal{H} (q)$ of the following type 
\begin{equation}\label{cohomologyoperatorbis}
 \ c(x_1, x_2, x_3)  \mapsto     c (x_2,x_3)  - c(x_1,x_3) + c (x_1,x_2) - c ([x_1,x_2],x_3) - c ([x_2,x_3],x_1) +  c ([x_1,x_3],x_2).\end{equation}
 Of course, a $2$-cochain $c$ is a $2$-cocycle if   $\mathrm{d}_2$ is constant to $0$ on $c$.

In this way we have just introduced the low dimensional Chevalley-Eilenberg cohomology for $\mathcal{H}(q)$, concerning the case of dimension two. A similar description holds for dimension one, or for higher dimensions, but we are going to focus only in the 2-dimensional case.  In line with Definition \ref{ce1} when we studied Lie algebras, we now consider the set   $Z^2(\mathcal{H}(q),V)$ of all $2$-cocycles on $\mathcal{H}(q)$, noting that it forms an abelian $q$-deformed Heisenberg algebra. Note that
\begin{equation}
Z^2(\mathcal{H}(q),V) = \ker \ \mathrm{d}_2  \ \cap \  C^2(\mathcal{H}(q),V).
\end{equation}
A $2$-cochain $c$ is called $2$-\textit{coboundary} if $c = \mathrm{d}_2 \ c'$ for some other ($1$)-cochain $c'$.  The set  $B^2(\mathcal{H}(q),V)$ of all $1$-coboundaries on $\mathcal{H}(q)$ forms an abelian $q$-deformed Heisenberg subalgebra of $Z^2(\mathcal{H}(q),V)$ 
\begin{equation}
B^2(\mathcal{H}(q),V) = \mathrm{im} \ \mathrm{d}_3  \ \cap \  C^2(\mathcal{H}(q),V)
\end{equation}
and so we arrive to the  $2$-\textit{cohomology} with complex coefficients of $\mathcal{H}(q)$, which is defined by 
\begin{equation}\label{qschur}
H^2(\mathcal{H}(q),V) = \frac{Z^2(\mathcal{H}(q),V)}{B^2(\mathcal{H}(q),V)}.
\end{equation}

\begin{remark}\label{qderivation2}  The \textit{Schur multiplier} $M(\mathcal{H}(q))$  of $\mathcal{H}(q)$, that is, the Schur multiplier of a $q$-deformed Heisenberg algebra on the complex field,  is exactly \eqref{qschur}.
This  is an abelian $q$-deformed Heisenberg algebra which becomes exactly what we have in Remark \ref{ce3} when $q=1$. Also here we have  strong relations with extension theory, since one can show that  variations of \cite[Theorem 7.6.3]{weibel} applies to $\mathcal{H}(q)$, that is, there exists a bijective correspondence between an element of $M(\mathcal{H}(q)) $ and an appropriate class of equivalence of extensions of $\mathcal{H}(q)$. Also here the idea is that 
 $M(\mathcal{H}(q)) $ allows us to control all possible  extensions of $\mathcal{H}(q)$.
\end{remark}

\bigskip
\bigskip
\bigskip

%\subsection*{Data accessibility statement}

%This work does not have any experimental data.

%\subsection*{Competing interests statement}

%We have no competing interests.

%\section*{Authors' contributions}

%\section*{Acknowledgements}

%\section*{Funding statement}


\newpage



\begin{thebibliography}{99}


%\bibitem{truncCCR} H. A. Buchdahl, {\em Concerning a kind of truncated quantized linear harmonic oscillator}, Amer. Jour. of Phys., {\bf 35}, 210 (1967);  B. Bagchi, S.N. Biswas, A. Khare, P.K. Roy,   {\em Truncated harmonic oscillator and parasupersymmetric quantum mechanics}, Pramana {\bf 49}, (2), 199-204 (1997)


\bibitem{kuw} A. Alazemi, A. Al-Dhafeer and L.A.M. Hanna,  On faithful matrix representations of $q$-deformed models in
quantum optics, \textit{Int.  J.  Math.  Math.  Sci. }
Volume 2022, Article ID 6737287, 8 pages

\bibitem{ancogoze1} J. Ancochea-Bermudez and M. Goze, 
  On the varieties of nilpotent Lie algebras of dimension 7 and 8,
   \textit{ J.\ Pure Appl.\ Algebra} \textbf{77} (1992), 131--140.


%\bibitem{ancogoze2} J. Ancochea-Bermudez and M. Goze,  Classification of nilpotent complex Lie algebras of dimension 7, \textit{Arch. Math.} \textbf{52} (1989), 175--185.

\bibitem{aitbook}  J.-P. Antoine, A. Inoue and  C. Trapani, {\it Partial $*-$algebras and their	operator realizations}, Kluwer, Dordrecht, 2002.


%\bibitem{bagadd1} F. Bagarello, M. Lattuca, R. Passante, L. Rizzuto, S. Spagnolo, {\em A Non-Hermitian Hamiltonian for a Modulated Jaynes-Cummings Model with ${\mathcal PT}$ Symmetry},  Phys. Rev. A, {\bf 91},  4, 042134 (2015)

%\bibitem{bagadd2} F. Bagarello, {\em Appearances of pseudo-bosons from Black-Scholes equation},  J. Math. Phys., {\bf 57}, 043504 (2016)

%\bibitem{bagadd3} F. Bagarello, {\em $kq$-representation for pseudo-bosons, and completeness of bi-coherent states},  JMAA, {\bf 450}, 631-643 (2017)


\bibitem{bag2022book}F. Bagarello, \textit{Pseudo-bosons and their coherent states}, Springer, Cham, 2022. 


%\bibitem{bagadd4} F. Bagarello,  Applications of topological $*$-algebras of unbounded operators to modified quons, \textit{Nuovo Cimento B} {\bf 117}  (2002), 593--611.


%\bibitem{bagcs1} F. Bagarello {\em Pseudo-bosons, Riesz bases and coherent states},   J. Math. Phys., {\bf 50}, 023531 (2010) (10pg)


%\bibitem{abg} F. Bagarello, S. T. Ali, J. P. Gazeau,  $\mathcal{D}$-pseudo-bosons, Complex Hermite Polynomials, and Integral Quantization, SIGMA, \textbf{11} (2015), 078,  23 pages



\bibitem{bagcohsta} F. Bagarello,  Intertwining operators for non self-adjoint Hamiltonians and bicoherent states,  \textit{J. Math. Phys.} {\bf 57} (2016), 103501.


%\bibitem{bagcs3} F. Bagarello, {\em Pseudo-bosons and Riesz bi-coherent states},  in {\em Geometric Methods in Physics},  Kielanowski, P., Ali, S.T., Bieliavsky, P., Odzijewicz, A., Schlichenmaier, M., Voronov, T. (Eds.), Springer   (2016)

\bibitem{baginbook} F. Bagarello,  Deformed canonical (anti-)commutation relations and non hermitian Hamiltonians, in \textit{Non-selfadjoint operators in quantum physics: Mathematical aspects}, M. Znojil et al. (Eds.), John Wiley and Sons, Hoboken,  2015,  pp.121--188.

\bibitem{bagquons} F. Bagarello,  Deformed quons and bi-coherent states,   \textit{Proc. Roy. Soc. A} {\bf 473} (2017), 20170049. 


\bibitem{bagrus2018} F. Bagarello and F. G. Russo, A description of pseudo-bosons in terms of  nilpotent Lie algebras,  \textit{J.. Geom.  Phys.} {\bf 125} (2018), 1--11. 


\bibitem{bagrus2019}F. Bagarello and F.G. Russo, On the presence of families of pseudo-bosons in nilpotent Lie algebras of arbitrary corank, \textit{J. Geom. Physics} \textbf{137} (2019), 124--131.

\bibitem{bagrus2020} F. Bagarello and F.G. Russo, Realization of Lie algebras of high dimension via pseudo-bosonic operators, \textit{J. Lie Theory} \textbf{30} (2020), 925--938.


%\bibitem{bagrev2007} F. Bagarello,  Algebras of unbounded operators and physical applications: a survey,  {\em Reviews in Math. Phys}  {\bf 19} (2007),  231--272.

%\bibitem{tirao}J. Barrionuevo, P. Tirao and D. Sulca, Deformations and rigidity in varieties of Lie algebras,\textit{J. Pure  Applied Algebra} \textbf{ 227} (2023), 107217.


\bibitem{bender1}  C. Bender,  Making sense of non-Hermitian Hamiltonians, \textit{Rep. Progr.  Phys.} {\bf 70} (2007),  947--1018.

\bibitem{bz} M. Bo\.zejko, R. Speicher, An example of a generalized {B}rownian motion, \textit{Comm. Math. Phys.}, \textbf{137} (1991), 519--531 

\bibitem{burde1} D. Burde,    Degenerations of 7-dimensional nilpotent Lie algebras,  \textit{ Comm. Algebra} \textbf{33} (2005),  1259--1277.  
   
\bibitem{burde2} D. Burde and C. Steinhoff,   Classification of orbit closures of 4-dimensional complex Lie algebras, \textit{J. Algebra} \textbf{214} (1999),  729--739.     

\bibitem{cantuba} R.R.S. Cantuba,  Lie polynomials in $q$-deformed Heisenberg algebras, \textit{J. Algebra} \textbf{522} (2019),  101--123.


\bibitem{carles}R. Carles, Sur la structure des alg\'ebres de Lie rigides, \textit{Ann. Inst. Fourier (Grenoble)} \textbf{34} (1984), 65--82.

\bibitem{chri} O. Christensen, \textit{An introduction to frames and Riesz bases}, Birkh\"auser, Boston, 2003.


%\bibitem{curado} E. M. F. Curado and M. A. Rego-Monteiro,  Multi-parametric deformed Heisenberg algebras: a route to complexity,  \textit{J. Phys. A} {\bf 34} (2001), 3253--3264. 


\bibitem{dixmier}J. Dixmier, Cohomologie de alg\'ebre de Lie nilpotentes, \textit{Acta  Sci. Math. Szeged} \textbf{16} (1955), 246--250.

%\bibitem{baz} M. El Baz and Y. Hassouni,  Deformed exterior algebra, quons and their coherent states, \textit{Int. J. Modern Physics A} {\bf 18} (2003),  3015--3040.


\bibitem{tripf} O. Cherbal, M. Drir, M. Maamache  and D. A. Trifonov,  Fermionic coherent states for pseudo-Hermitian two-level systems, \textit{J. Phys. A} {\bf 40}  (2007), 1835--1844.


\bibitem{eremel} V.V. Eremin and A.A. Meldianov,  The q-deformed harmonic oscillator, coherent states and the uncertainty relation, \textit{Theor. Math. Phys.} {\bf 147} (2006), 709--715.

\bibitem{fiv} D.I. Fivel,  Interpolation between Fermi and Bose
statistics using generalized commutators, \textit{Phys. Rev. Lett.} {\bf 65} (1990), 3361--3364, ; Erratum, \textit{Phys. Rev. Lett.} {\bf 69} (1992), 2020.


\bibitem{gerstenhaber1} M. Gerstenhaber, 
   On the deformations of rings and algebras, 
   \textit{Ann. Math.} \textbf{74} (1964),  59--103.

\bibitem{gerstenhaber2}M. Gerstenhaber, On the Deformation of Rings and Algebras: II,  \textit{Ann. Math.} \textbf{84} (1966),  1--19.

\bibitem{gerstenhaber3}M. Gerstenhaber, On the Deformation of Rings and Algebras: III,  \textit{Ann. Math.} \textbf{88} (1968), 1--34.


\bibitem{gerstenhaber4}M. Gerstenhaber, On the Deformation of Rings and Algebras: IV,  \textit{Ann. Math.} \textbf{99} (1974), 257--276.

%\bibitem{gerstenhaber2}M.  Gerstenhaber  and S. Schack, Relative Hochschild cohomology, rigid algebras, and the Bockstein,  \textit{ J. Pure Appl. Algebra} \textbf{43} (1986), 53--74.    

%\bibitem{kar}  G. Ghosh and T.K. Kar, Coherent states for quons, \textit{J. Phys. A} {\bf 29} (1996), 125--131.



\bibitem{gokh}M. Goze and Y.  Khakimdjanov, \textit{
Nilpotent Lie algebras},  Kluwer Academic Publishers, Dordrecht, 1996. 


\bibitem{gre} O.W. Greenberg,  Particles with small violations of Fermi or Bose statistics, \textit{Phys. Rev. D} {\bf 43} (1991), 4111--4120.

\bibitem{grunewald1} F. Grunewald and J. O'Halloran,  Deformations of Lie algebras,  \textit{J. Algebra} {\bf 162} (1993), 210--224. 

\bibitem{grunewald2} F. Grunewald and J. O'Halloran, 
  A characterization of orbit closure and applications,  \textit{  J. Algebra} \textbf{116} (1988),  163--175.


%\bibitem{bagbook}  F. Bagarello, J. P. Gazeau, F. H. Szafraniec and M. Znojil Eds., {\em Non-selfadjoint operators in quantum physics: Mathematical aspects},  Wiley  (2015)

\bibitem{hartshorne}R. Hartshorne, \textit{Algebraic Geometry}, Springer, Berlin, 1977.

\bibitem{sergey1}J. T. Hartwig, D. Larsson and S. Silvestrov, Deformations of Lie algebras using $\sigma$-derivations, \textit{J. Algebra} \textbf{265} (2006), 314--361.




\bibitem{sergey2} L. Hellstr\"om and S. Silvestrov,  \textit{Commuting elements in q-deformed Heisenberg algebras}, World Scientific, Singapore, 2000.

\bibitem{herrera1}J.F. Herrera-Granada and P. Tirao,  The Grunewald-O'Halloran Conjecture for nilpotent Lie algebras of rank $\ge 1$, \textit{ Comm.  Algebra} \textbf{44}, 2180--2192. 

\bibitem{herrera2}J.F. Herrera-Granada. O Marquez and S. Vera, Degenerations to filiform Lie algebras of dimension 9,  \textit{Comm. Algebra}  \textbf{50} (2022),  836--847.

\bibitem{herrera3}J.F. Herrera-Granada and P. Tirao, Filiform Lie algebras of dimension 8 as degenerations,  \textit{J. Algebra   Appl.}  \textbf{13} (2014), 1350144.



\bibitem{hofmor}K.H. Hofmann and S.A. Morris, \textit{The structure of compact groups}, de Gruyter, Berlin, 2023.

\bibitem{qccr} P.E.T. Jorgensen, L. M. Schmitt and R.F. Werner, $q$-canonical commutation relations and stability of the Cuntz algebra, \textit{Pacific J. Math.} \textbf{165}   (1994), 131--151.
 
 
\bibitem{kont1} M. Kontsevich,  Deformation quantization of Poisson manifolds, \textit{Letters  Math.  Phys.}  \textbf{66} (2003),  157--216.
 
\bibitem{kont2} M. Kontsevich, Formality conjecture. Deformation theory and symplectic geometry, In: Ascona 1996, D. Sternheimer et al. (eds), Math. Phys. Stud. 20, Kluwer, Dordrecht, 1997, pp.139--156. 
 
 
 
\bibitem{moh} R.N. Mohapatra,  Infinite statistics and a possible small violation of the Pauli principle, \textit{Phys. Lett. B} {\bf 242} (1990), 407--411.



\bibitem{morozov}V.V. Morozov, Classification of nilpotent Lie algebras of sixth order,\textit{ Izv. Vyss. Ucebn. Zaved. Mat.} \textbf{4} (1958), 161--171.


%\bibitem{mosta} A. Mostafazadeh,  Pseudo-Hermitian representation of quantum mechanics, \textit{Int. J. Geom. Methods Mod. Phys.} {\bf 7} (2010), 1191--1306.



%\bibitem{specissue} C. Bender, A. Fring, U. G\"unther, H. Jones Eds., {\em Special issue on quantum physics with non-Hermitian operators}, J. Phys. A: Math. and Theor., vol. {\bf 45}, N. 44 (2012)

%\bibitem{baginbagbook} F. Bagarello, {\em Deformed canonical (anti-)commutation relations and non hermitian hamiltonians}, in {Non-selfadjoint operators in quantum physics: Mathematical aspects}, F. Bagarello, J. P. Gazeau, F. H. Szafraniec and M. Znojil Eds., John Wiley and Sons Eds.  (2015)


\bibitem{rem1}E. Remm, Rigid Lie algebras and algebraicity,
\textit{Rev. Roum. Math. Pures Appl.} \textbf{65} (2020), 491--510. 


\bibitem{seeley1}C. Seeley, Degenerations of 6-dimensional nilpotent Lie algebras over $\mathbb{C}$, \textit{Comm. Algebra} \textbf{18} (1990),  3493--3505.

\bibitem{seeley2} C. Seeley and S.S.-T. Yau, Variation of complex structures and variation of Lie algebras \textit{Invent. Math.} \textbf{99} (1990),  545--565.


\bibitem{sergey3}G. Sigurdsson and S. Silvestrov, Bosonic realizations of the colour Heisenberg Lie algebra, \textit{J. Nonlinear Math. Phys.} \textbf{13} (2006), 110--128. 

%\bibitem{bagquons1} F. Bagarello, {\em  Applications of Topological *-Algebras of Unbounded
%Operators to Modified Quons}, Il Nuovo Cimento B, {\bf 117}, No.5,
%593-611, (2002)



%\bibitem{sergey4} S. Silvestrov and L. Turowska, Representations of the q-deformed Lie algebra of the group of motions of the Euclidean plane, \textit{J. Funct. Anal.} \textbf{160} (1998), 79--114. 



\bibitem{sund}T. Skjelbred and T. Sund,  Sur la classification des alg\'ebres de Lie nilpotentes, \textit{C. R. Acad. Sci. Paris S\'er. A--B} \textbf{286} No. 5  (1978), A241--A242.



\bibitem{swanson} M.S. Swanson,  Transition elements for a non-Hermitian quadratic Hamiltonian, {\em J. Math. Phys.} {\bf 45} (2004), 585--601.





%\bibitem{jones} F. Jones, {Lebesgue integration on Euclidean space - Revised Edition}, Jones and Bartlett Eds., Sudbury (2001)

%\bibitem{kolfom} A. Kolmogorov and S. Fomine, {\em El\'ements de la th\'eorie des fonctions et de lanalyse fonctionelle}, Mir (1973)

\bibitem{snobwin}L. Snobl and P. Winternitz,
\textit{Classification and identification of Lie algebras}, CRM Monograph Series, vol. 33, American Mathematical Society, Providence, RI, 2014.




%\bibitem{trrev} C. Trapani, Quasi $*$-algebras of operators and their applications,  \textit{ Reviews Math. Phys.} \textbf{ 7} (1995),  1303--1332.


\bibitem{vergne} M. Vergne, Cohomologie des alg\`{e}bres de Lie nilpotentes, \textit{   Bull. Soc. Math. France} \textbf{98} (1970), 81--116. 


\bibitem{vN39} J. von Neumann,   On infinite direct products, \textit{Compositio Math.} {\bf 6}  (1939), 1--77 


\bibitem{weibel}C. Weibel, \textit{An introduction to homological algebra}, Cambridge University Press, Cambridge, 1997.


\end{thebibliography}
\end{document}